\pdfoutput=1

\documentclass[12pt,english]{article}
\usepackage[T1]{fontenc}
\usepackage[utf8]{inputenc}
\usepackage{geometry}
\usepackage{fancyhdr}
\usepackage{color}
\usepackage{babel}
\usepackage{amsthm}
\usepackage{amsmath,accents}

\usepackage[unicode=true,pdfusetitle,
 bookmarks=true,bookmarksnumbered=true,bookmarksopen=true,bookmarksopenlevel=3,
 breaklinks=false,pdfborder={0 0 1},backref=false,colorlinks=true]
 {hyperref}
\usepackage{bookmark}

\makeatletter
\numberwithin{equation}{section}
\theoremstyle{plain}
\newtheorem{thm}{\protect\theoremname}
  \theoremstyle{plain}
  \newtheorem{conjecture}[thm]{\protect\conjecturename}
  \theoremstyle{remark}
  \newtheorem*{rem*}{\protect\remarkname}
  \theoremstyle{remark}
  \newtheorem{rem}[thm]{\protect\remarkname}
  \theoremstyle{plain}
  \newtheorem{prop}[thm]{\protect\propositionname}

\numberwithin{equation}{section}
\numberwithin{figure}{section}

\usepackage{chngcntr}
\counterwithout{figure}{section}
\usepackage[footnotesize]{caption}

\usepackage{csquotes}

\usepackage[backend=bibtex,
sorting=nyt,
backref=true,
sortcites=true,
firstinits=true,
language=english,
maxbibnames=99,
isbn=false,url=false,doi=true
]{biblatex}

\renewbibmacro{in:}{%
  \ifentrytype{article}{}{\printtext{\bibstring{in}\intitlepunct}}}

\DefineBibliographyStrings{english}{%
    backrefpage  = {cit. on p.}, 
    backrefpages = {cit. on pp.} 
}

\DeclareFieldFormat{journaltitle}{\mkbibemph{#1},} 
\DeclareFieldFormat[inbook,thesis]{title}{\mkbibemph{#1}\addperiod} 
\DeclareFieldFormat[article]{title}{#1} 
\DeclareFieldFormat[article]{volume}{\textbf{#1}\addcolon\space} 

\usepackage{etex}
\usepackage{empheq}

\usepackage{txfonts}
\renewcommand{\widehat}{\hat}
\renewcommand{\widetilde}{\tilde}

\usepackage[activate={true,nocompatibility},final,verbose=silent,tracking=true,kerning=true,spacing=true,factor=1100,stretch=20,shrink=20]{microtype}
\UseMicrotypeSet[protrusion]{basictext}
\microtypecontext{spacing=nonfrench}
\DisableLigatures[N]{encoding = *, family = * }

\DeclareMathSizes{12}{12}{8}{6}
\usepackage{array}
\newcolumntype{C}{>{$\displaystyle} c <{$}}   

\usepackage{caption}
\usepackage{graphicx}
\usepackage{xcolor}
\graphicspath{{figures/}}

\usepackage{transparent}

\renewcommand{\geq}{\geqslant}
\renewcommand{\leq}{\leqslant}

\renewcommand{\Re}{\operatorname{Re}}

\newcommand{\tg}{\operatorname{tg}}

\DeclareMathOperator*{\res}{res}
\def\Ai{{\rm Ai \,}}
\def\bigO{\mathcal{O}}
\newcommand{\restr}[2]{{
  \left.\kern-\nulldelimiterspace 
  #1 
  \vphantom{\big|} 
  \right|_{#2} 
  }}

  \addto\captionsenglish{\renewcommand{\propositionname}{Proposition}}
  \addto\captionsenglish{\renewcommand{\remarkname}{Remark}}
  \addto\captionsenglish{\renewcommand{\theoremname}{Theorem}}
  \providecommand{\propositionname}{Proposition}
  \providecommand{\remarkname}{Remark}
\providecommand{\theoremname}{Theorem}

  \providecommand{\conjecturename}{Conjecture}
  \providecommand{\propositionname}{Proposition}
  \providecommand{\remarkname}{Remark}
\providecommand{\theoremname}{Theorem}

\makeatother

  \providecommand{\conjecturename}{Conjecture}
  \providecommand{\propositionname}{Proposition}
  \providecommand{\remarkname}{Remark}
\providecommand{\theoremname}{Theorem}

\usepackage{esint}

\begin{document}

\date{\today}

\title{Hankel determinant and orthogonal polynomials for a Gaussian weight
with a discontinuity at the edge}

\author{A. Bogatskiy\thanks{Department of Higher Mathematics and Mathematical Physics, Department
of Physics, Saint-Petersburg State University. Saint-Petersburg, Russia.}, T. Claeys\thanks{Institut de Recherche en Mathématique et Physique, Université catholique
de Louvain, Chemin du Cyclotron 2, B-1348 Louvain-La-Neuve, Belgium}, A. Its\thanks{Department of Mathematical Sciences, Indiana University -- Purdue
University Indianapolis. Indianapolis, IN 46202-3216, USA.}}
\maketitle
\begin{abstract}
We compute asymptotics for Hankel determinants and orthogonal polynomials
with respect to a discontinuous Gaussian weight, in a critical regime
where the discontinuity is close to the edge of the associated equilibrium
measure support. Their behavior is described in terms of the Ablowitz-Segur
family of solutions to the Painlevé II equation. Our results complement
the ones in \cite{XuZhao11}. As consequences of our results, we conjecture
asymptotics for an Airy kernel Fredholm determinant and total integral
identities for Painlevé II transcendents, and we also prove a new
result on the poles of the Ablowitz-Segur solutions to the Painlevé
II equation. We also highlight applications of our results in random
matrix theory. 
\end{abstract}

\section{Introduction}

Consider the Hankel determinant, 
\begin{equation}
H_{n}(\lambda_{0},\beta)=\det\left(\int_{-\infty}^{\infty}x^{j+k}w(x)\mathrm{d}x\right)_{j,k=0}^{n-1}=\frac{1}{n!}\dotsint_{-\infty}^{\infty}\prod_{i<j}(x_{i}-x_{j})^{2}\prod_{k=1}^{n}w(x_{k})\mathrm{d}x_{k},\label{det}
\end{equation}
with respect to a discontinuous Gaussian weight of the form 
\begin{equation}
w(x)=e^{-x^{2}}\times\begin{cases}
e^{\pi i\beta}, & x<\lambda_{0}\\
e^{-\pi i\beta}, & x\geq\lambda_{0}
\end{cases},\quad\Re\beta\in\left(-\frac{1}{2},\frac{1}{2}\right],\,\lambda_{0}\in\mathbb{R}.\label{weight}
\end{equation}
The weight is periodic in $\beta$ and we can restrict to the case
$-1/2<\Re\beta\leq1/2$ without loss of generality. If $\beta$ is
purely imaginary, the weight is positive.

We also consider the monic orthogonal polynomials $p_{n}$ of degree
$n$ with respect to the weight $w(x)$ on the real line, defined
by the orthogonality conditions 
\begin{equation}
\int_{-\infty}^{\infty}p_{n}(x)p_{m}(x)w(x)\mathrm{d}x=h_{n}\delta_{nm},\qquad h_{n}=h_{n}(\lambda_{0},\beta).\label{pol}
\end{equation}
Those are connected to the Hankel determinant $H_{n}$ by the well-known
identity $H_{n}(\lambda_{0},\beta)=\prod_{k=0}^{n-1}h_{k}(\lambda_{0},\beta)$.
We denote by $R_{n}=R_{n}(\lambda_{0},\beta)$ and $Q_{n}=Q_{n}(\lambda_{0},\beta)$
the recurrence coefficients in the three-term recurrence relation
\begin{equation}
xp_{n}(x)=p_{n+1}(x)+Q_{n}p_{n}(x)+R_{n}p_{n-1}(x).\label{recur}
\end{equation}

The question which we are concerned with in this paper is the large
$n$ behavior of the Hankel determinants $H_{n}$, the polynomials
$p_{n}(x)$, and their recurrence coefficients $R_{n}$ and $Q_{n}$,
in the regime where the point of discontinuity $\lambda_{0}$ behaves
like $\sqrt{2n}$. They can asymptotically be expressed in terms of
the Ablowitz-Segur solutions to the Painlevé II equation. As important
by-products of the asymptotics for the Hankel determinants, we also
conjecture so-called large gap asymptotics for an Airy kernel Fredholm
determinant and total integral identities for the Ablowitz-Segur solutions
of the Painlevé II equation. Relying on a result of \cite{XuZhao11},
we in addition prove a new result about the poles for those Painlevé
transcendents.

If we let $\lambda_{0}=\lambda\sqrt{2n}$, the large $n$ asymptotics
of the orthogonal polynomials (\ref{pol}), the recurrence coefficients
(\ref{recur}), and the Hankel determinant (\ref{det}) depend dramatically
on whether $|\lambda|<1$ or $|\lambda|>1$, i.e. whether the jump
location $\lambda_{0}$ is inside or outside of the support $[-\sqrt{2n},\sqrt{2n}]$
of the equilibrium measure with Gaussian external field. In the case
$|\lambda|>1$, all the objects of interest behave effectively as
they do for the pure Gaussian weight (i.e., the case where we formally
set $\lambda_{0}=+\infty$); the discontinuity yields an exponentially
small correction only \cite{Johansson98}. In the case $|\lambda|<1$,
the situation is different; the discontinuity of the weight becomes
strongly visible in the large $n$ behavior of the orthogonal polynomials,
the recurrence coefficients, and the Hankel determinant \cite{ItsKras08}.
For the Hankel determinant, it was proved in \cite[equation (1.5)]{ItsKras08}
that 
\begin{multline}
H_{n}(\lambda_{0},\beta)=H_{n}\left(\lambda_{0},0\right)\,G(1+\beta)G(1-\beta)(1-\lambda^{2})^{-3\beta^{2}/2}(8n)^{-\beta^{2}}\enskip\times\\
\times\enskip\exp\left(2in\beta\left(\arcsin\lambda+\lambda\sqrt{1-\lambda^{2}}\right)\right)\left(1+\bigO\left(\frac{\log n}{n^{1-4|\Re\beta|}}\right)\right),\;\left|\Re\beta\right|<\frac{1}{4},\label{as Hankel noncrit}
\end{multline}
as $n\to\infty$, uniformly for $\lambda$ in compact subsets of $(-1,1)$.
Here $G$ is the Barnes' $G$-function, and 
\begin{equation}
H_{n}\left(\lambda_{0},0\right)=(2\pi)^{n/2}2^{-n^{2}/2}\prod_{k=1}^{n-1}k!\label{HnGUE}
\end{equation}
denotes the Hankel determinant corresponding to the pure Gaussian
weight $e^{-x^{2}}$. Asymptotics for the recurrence coefficients
$Q_{n}$ and $R_{n}$ in the case $-1<\lambda<1$ are also given in
\cite{ItsKras08}.

In this paper, we analyze the transition regime where the point $\lambda_{0}$
of discontinuity of the weight is (relatively) close to $\sqrt{2n}$.
More precisely we let 
\begin{equation}
\lambda_{0}=\lambda\sqrt{2n},\qquad\lambda=1+\frac{t}{2}n^{-2/3},\label{scaling lambda0}
\end{equation}
where $t\in\mathbb{R}$. We will see that the asymptotic behavior
of $H_{n}$, $p_{n}$, $R_{n}$, and $Q_{n}$ depends in a non-trivial
way on the parameter $t$ in \eqref{scaling lambda0}. The asymptotic
behavior is described in terms of a family of solutions to the Painlevé
II equation 
\begin{equation}
u_{tt}=tu+2u^{3},\label{P2}
\end{equation}
with the asymptotic behavior 
\begin{equation}
u(t;\kappa)\sim\kappa{\rm Ai}(t),\qquad t\to+\infty,\label{u+}
\end{equation}
where ${\rm Ai}$ denotes the Airy function, and 
\begin{equation}
u(t;\kappa)=\frac{1}{(-t)^{1/4}}\sqrt{2i\beta}\sin\phi(t;\beta)+\bigO\left(\frac{1}{t^{2-3|\Re\beta|}}\right),\quad t\to-\infty,\label{u-}
\end{equation}
with 
\begin{equation}
\phi(t;\beta)=-\frac{\pi}{4}-i\log\frac{\Gamma(-\beta)}{\Gamma(\beta)}+\frac{2}{3}(-t)^{3/2}-\frac{3}{2}i\beta\log\left(-t\right)-3i\beta\log2,\qquad\kappa^{2}=1-e^{-2\pi i\beta},\qquad|\Re\beta|<\frac{1}{2}.\label{theta}
\end{equation}
For $0<\kappa<1$, these solutions are known as the Ablowitz-Segur
solutions \cite{AblowSegu77} of the second Painlevé equation. They
are uniquely characterized either by \eqref{u+} or by \eqref{u-}.
Moreover, it is known that $u(t;\kappa)$ has no singularities for
$t$ on the real line if $\kappa\in i\mathbb{R}$ or if $|\kappa|<1$.
For $\kappa\in\mathbb{R}\setminus[-1,1]$, or equivalently $|\Re\beta|=1/2$,
it is known that $u(\tau;\kappa)$ does have real poles \cite{Bertola12}.
Relying on a result from \cite{XuZhao11}, we will prove the following
result, stating that $u$ has no real poles for any $\kappa\in\mathbb{C}\setminus((-\infty,-1]\cup[1,+\infty))$,
or equivalently for any $\beta$ with $|\Re\beta|<1/2$. 
\begin{thm}
\label{theorem poles} Let $u(t;\kappa)$ be the solution to the Painlevé
II equation \eqref{P2} characterized by \eqref{u+}. If $\kappa\in\mathbb{C}\setminus((-\infty,-1]\cup[1,+\infty))$
is fixed, then $u(t;\kappa)$ has no poles at real values of $t$. 
\end{thm}
In the case $\kappa=0$, we simply have $u(t;\kappa)=0$; the unique
Painlevé II solution satisfying \eqref{u+} with $\kappa=\pm1$ (which
means formally that $\beta=-i\infty$) is known as the Hastings-McLeod
solution.

The function $y(t;\beta)=u(t;\kappa)^{2}$ solves the Painlevé XXXIV
equation 
\begin{equation}
y_{tt}=4y^{2}+2ty+\frac{(y_{t})^{2}}{2y}.\label{P34}
\end{equation}

The function $y(t,\beta)$ and equation \eqref{P34} are, in fact,
the objects which directly appear in our double scaling analysis of
$H_{n}$, $p_{n}$, $R_{n}$ and $Q_{n}$. Our next result describes
the asymptotics or the Hankel determinants $H_{n}(\lambda_{0},\beta)$. 
\begin{thm}
\label{theorem hankel} Let $|\Re\beta|<1/2$ and let $H_{n}(\lambda_{0},\beta)$
be the Hankel determinant (\ref{det}) corresponding to the weight
(\ref{weight}), with $\lambda_{0}$ given by (\ref{scaling lambda0}).
If $\kappa^{2}=1-e^{-2\pi i\beta}$, we have 
\begin{equation}
H_{n}\bigl(\lambda_{0},\beta)=e^{i\pi\beta n}H_{n}\left(\lambda_{0},0\right)\exp\left(-\int_{t}^{\infty}(\tau-t)u(\tau;\kappa)^{2}\mathrm{d}\tau\right)(1+o(1)),\qquad n\to\infty,\label{hankelas}
\end{equation}
uniformly for $t\in[-M,\infty)$ for any $M>0$ and for $\beta$ in
compact subsets of $|\Re\beta|<1/2$, where $H_{n}\left(\lambda_{0},0\right)$
is given in \eqref{HnGUE}. 
\end{thm}
Theorem \ref{theorem hankel} has two consequences which are not directly
related to the Hankel determinants or orthogonal polynomials studied
in this paper, but which are of independent interest. To describe
them, we note first that the exponential in \eqref{hankelas} can
be recognized as the Tracy-Widom formula for the Fredholm determinant
$\det\left(1-\kappa^{2}\restr{K_{\mathrm{Ai}}}{[t,+\infty)}\right)$,
where $\restr{K_{\mathrm{Ai}}}{[t,+\infty)}$ is the integral operator
with kernel 
\begin{equation}
K_{\Ai}(x,y)=\frac{\Ai(x)\Ai'(y)-\Ai(y)\Ai'(x)}{x-y}
\end{equation}
acting on $[t,+\infty)$. Indeed, it was shown in \cite{TracyWido94}
that 
\begin{equation}
\det\left(1-\kappa^{2}\restr{K_{\mathrm{Ai}}}{[t,+\infty)}\right)=\exp\left(-\int_{t}^{\infty}(\tau-t)u(\tau;\kappa)^{2}\mathrm{d}\tau\right).\label{TW}
\end{equation}

This observation, together with a strengthened version of the Hankel
determinant asymptotics \eqref{as Hankel noncrit}, allows us to formulate
the following conjecture about the $t\to-\infty$ asymptotics of $\det\left(1-\kappa^{2}\restr{K_{\mathrm{Ai}}}{[t,+\infty)}\right)$. 
\begin{conjecture}
\label{conj Airy}Let $\kappa\in\mathbb{C}\setminus((-\infty,-1]\cup[1,+\infty))$
and define $\beta$ by 
\begin{equation}
\kappa^{2}=1-e^{-2\pi i\beta},\qquad|\Re\beta|<1/4.\label{def beta}
\end{equation}
As $t\to-\infty$, we have 
\begin{equation}
\log\det\left(1-\kappa^{2}\restr{K_{\mathrm{Ai}}}{[t,+\infty)}\right)=-\frac{4}{3}i\beta\left(-t\right)^{3/2}-\frac{3}{2}\beta^{2}\log\left(-t\right)+\log\left(G\left(1+\beta\right)G\left(1-\beta\right)\right)-3\beta^{2}\log2+o(1),\label{det Airy}
\end{equation}
or equivalently in form of a total integral identity 
\begin{equation}
\lim_{t\to-\infty}\left(-\int_{t}^{\infty}(\tau-t)u(\tau;\kappa)^{2}d\tau+\frac{4}{3}i\beta\left(-t\right)^{3/2}+\frac{3}{2}\beta^{2}\log\left(-t\right)\right)=\log\left(G\left(1+\beta\right)G\left(1-\beta\right)\right)-3\beta^{2}\log2.\label{total integral u}
\end{equation}
\end{conjecture}
\begin{rem*}
Similar asymptotics for the Airy kernel determinant in the case $\kappa=1$
were proved in \cite{DeiItsKra08,BaiBucDiF08}: we then have 
\begin{equation}
\log\det\left(I-\restr{K_{\mathrm{Ai}}}{[t,+\infty)}\right)=\frac{t^{3}}{12}-\frac{1}{8}\log\left|t\right|+c_{0}+\mathcal{O}\left(t^{-3}\right),\qquad t\to-\infty,\label{eq: k 1 airy-det asymptotic}
\end{equation}
where $c_{0}=\log2/24+\zeta'(-1)$ and $\zeta$ is the Riemann $\zeta$-function.
As $\kappa\to1$, it was shown recently in \cite{Bothner15} that
\begin{equation}
\log\det\left(I-\kappa^{2}\restr{K_{\mathrm{Ai}}}{[t,+\infty)}\right)=\frac{t^{3}}{12}-\frac{1}{8}\log\left|t\right|+c_{0}+o\left(1\right),\qquad t\to-\infty,
\end{equation}
as long as $\kappa<1$, and $\kappa\to1$ sufficiently rapidly so
that 
\begin{equation}
-\frac{\log(1-\kappa^{2})}{(-t)^{3/2}}>\frac{2\sqrt{2}}{3}.
\end{equation}
The total integrals of different expressions involving the second
Painlevé transcendent were studied in \cite{BaiBuDiIt09}. The integral
\eqref{total integral u} does not belong to the type which can be
handled by the technique of \cite{BaiBuDiIt09}. Indeed, like the
similar integral corresponding to equation \eqref{eq: k 1 airy-det asymptotic},
the integral in \eqref{total integral u} belongs to the third, most
difficult type of total integrals of Painlevé functions as classified
in the end of Section 6 of \cite{BaiBuDiIt09}. This means that the
evaluation of this integral goes beyond the analysis of the Riemann-Hilbert
problem corresponding to the Ablowitz-Segur Painlevé II transcendent.
As we already indicated, the proof of \eqref{total integral u} can
be achieved via an improvement of the error term in \eqref{as Hankel noncrit}.
Another possibility is to use certain differential identities for
the Airy determinant in \eqref{TW} with respect to $\kappa$. We
intend to consider these issues in our next publication.

Additionally, the asymptotics of the PXXXIV transcendent $y\left(t;\beta\right)=u\left(t;\kappa\right)^{2}$
as $t\to-\infty$ can be calculated directly by the same method as
the ones for $t\to+\infty$. We will not present this computation
since it is mostly identical to the one in \cite{ItsKras08}, and
alternatively this asymptotics can be obtained using the connection
formulae for the Painlevé II equation \cite{Kapaev92}. Moreover,
the following singular asymptotics take place when $\Re\beta=1/2$.\end{rem*}
\begin{thm}
\label{thm: u- as}Let $u(t;\kappa)$ be the solution to the Painlevé
II equation \eqref{P2} characterized by \eqref{u+} and let $\kappa^{2}=1-e^{-2i\pi\beta}=1+e^{2\pi\gamma}$,
$\beta=1/2+i\gamma$, $\gamma\in\mathbb{R}$. Then $y\left(t;\beta\right)=u\left(t;\kappa\right)^{2}$
is a solution to the Painlevé XXXIV equation \eqref{P34} and has
the following asymptotics as $t\to-\infty$, away from the zeros of
trigonometric functions appearing in the denominators: 
\begin{gather}
y\left(t;\frac{1}{2}+i\gamma\right)=\frac{-t}{\cos^{2}\widetilde{\phi}}+\frac{1}{\sqrt{-t}}\left(-\gamma+\frac{1}{2}\tg\widetilde{\phi}+\frac{2\gamma}{\cos^{2}\widetilde{\phi}}+\frac{3\left(12\gamma^{2}-1\right)\sin\widetilde{\phi}}{16\cos^{3}\widetilde{\phi}}\right)+\mathcal{O}\left(\frac{1}{t^{2}}\right),\label{eq:1.22}\\
\mbox{where }\widetilde{\phi}(t;\gamma)=\frac{2}{3}\left(-t\right)^{3/2}+\frac{3}{2}\gamma\log(-t)+3\gamma\log2-\arg\Gamma\left(\frac{1}{2}+i\gamma\right).
\end{gather}

\end{thm}
Asymptotics of this type in relation to the second Painlevé equation
have been previously obtained via different methods in \cite{Kapaev92}
and in \cite{BothnIts12}, but the second term is a new result of
the present work. This computation is based on an undressing procedure
adapted from \cite{BothnIts12}. Thus we will not present the derivation
of \eqref{eq:1.22} either.

For the recurrence coefficients $R_{n}$ and $Q_{n}$, we have the
following result, which was partially obtained before in \cite{XuZhao11},
see Remark \ref{remark XZ} below. 
\begin{thm}
\label{theorem rec} Let $R_{n}$ and $Q_{n}$ be the recurrence coefficients
defined in (\ref{recur}), associated to the orthogonal polynomials
with respect to the weight (\ref{weight}). Let $|\Re\beta|<1/2$
and let $\lambda_{0}$ be given by \eqref{scaling lambda0}. Then,
as $n\to\infty$, the recurrence coefficients have the following expansions,
\begin{equation}
R_{n}(\lambda_{0},\beta)=\frac{n}{2}-\frac{1}{2}u(t;\kappa)^{2}n^{1/3}+\mathcal{O}(1),\label{Rnas}
\end{equation}
and 
\begin{equation}
Q_{n}(\lambda_{0},\beta)=-\frac{1}{\sqrt{2}}u(t;\kappa)^{2}n^{-1/6}+\mathcal{O}\left(n^{-1/2}\right),\label{Qnas}
\end{equation}
uniformly for $t\in[-M,\infty]$ for any $M>0$, where $\kappa$ is
given by \eqref{def beta}. Additionally, we have the asymptotics
of the normalizing coefficients $h_{n}$:
\begin{equation}
h_{n}=\frac{\pi\sqrt{2n}n^{n}}{2^{n}e^{n}}e^{i\pi\beta}\left(1+n^{-1/3}v\left(t;\kappa\right)\mathrm{d}t+n^{-2/3}\frac{1}{2}\left(v\left(t;\kappa\right)^{2}-u\left(t;\kappa\right)^{2}\right)+\mathcal{O}\left(n^{-1}\right)\right),\label{eq: h-n asymptotics}
\end{equation}
where 
\[
v\left(t;\kappa\right)=\int_{t}^{\infty}u\left(\tau;\kappa\right)^{2}\mathrm{d}\tau.
\]
\end{thm}
\begin{rem}
The formal substitution, 
\[
t=-2(1-\lambda)n^{2/3}
\]
in the asymptotics for the recurrence coefficients transforms them,
with the help of the asymptotic expansion (\ref{u-}), into the non-critical
asymptotics obtained in \cite{ItsKras08}. This important fact indicates,
at least on the formal level, that the description of the transition
regime in the large $n$ behavior of the recurrence coefficients is
complete. 
\end{rem}

\begin{rem}
\label{remark XZ} The general form of \eqref{Rnas} and \eqref{Qnas}
was formally suggested in \cite{Its11} (together with the asymptotic
characterization of the Painlevé II function $u(t;\kappa)$) and it
was proved by Xu and Zhao in \cite{XuZhao11}. They obtained their
asymptotic expansions in terms of $\widehat{u}(t)=2^{1/3}u(2^{-1/3}\tau)^{2}$.
It was noted that this is a solution of a Painlevé XXXIV equation,
but no asymptotics for $\widehat{u}(t)$ as $t\to\pm\infty$ were
obtained, and thus the authors of \cite{XuZhao11} did not identify
$\widehat{u}$ in terms of the Ablowitz-Segur solution characterized
by \eqref{u+} or \eqref{u-}. In fact, assuming the matching of the
estimates (\ref{Rnas}) and (\ref{Qnas}) with the non-critical formulae
of \cite{ItsKras08}, asymptotics for $\widehat{u}(t)$ as $t\to-\infty$
were deduced heuristically. There is, however, no independent derivation
of it which is needed for the rigorous completion of the analysis
of the transition regime in question. The $+\infty$ - characterization
of the Painlevé transcendent $\widehat{u}(t)$, even heuristically,
is not given in \cite{XuZhao11}. 
\end{rem}
As an additional result, we also obtain an analog of the Plancherel-Rotach
asymptotics for classical Hermite polynomials \cite{Planchere29}. 
\begin{thm}
\label{theorem pn}Let $p_{n}(x)$ be the degree $n$ monic orthogonal
polynomial with respect to the weight (\ref{weight}), and let $\lambda_{0}$
be given by \eqref{scaling lambda0}. Let $\left|\Re\beta\right|<1/2$.
Then, as $n\rightarrow\infty$, 
\begin{equation}
p_{n}\left(\lambda_{0}\right)=\frac{\sqrt{2\pi}}{\kappa}\left(\frac{ne}{2}\right)^{n/2}n^{1/6}e^{tn^{1/3}}u\left(t;\kappa\right)\left(1+\mathcal{O}\left(n^{-1/3}\right)\right),\label{PR as}
\end{equation}
with $\kappa$ given by \eqref{def beta}. \end{thm}
\begin{rem}
Using the asymptotic behavior \eqref{u+} for $u$ as $t\to+\infty$,
\eqref{PR as} matches formally with the classical Plancherel-Rotach
asymptotics for the Hermite polynomials \cite{Planchere29}: 
\begin{equation}
p_{n}\left(\lambda_{0}\right)=\sqrt{2\pi}\left(\frac{ne}{2}\right)^{n/2}n^{1/6}e^{tn^{1/3}}\Ai\left(t\right)\left(1+\mathcal{O}\left(n^{-1/3}\right)\right),\qquad n\to\infty,\label{PR Hermite}
\end{equation}
where $p_{n}$ are the monic Hermite polynomials.

On the other hand, if we let $\beta\to0$, or equivalently $\kappa\to0$,
we have (see equations \eqref{sign} and \eqref{c Airy} below) that
\begin{equation}
u(t;\kappa)=0,\qquad\lim_{\kappa\to0}\frac{1}{\kappa}u(t;\kappa)=\Ai(t),\label{u kappa0}
\end{equation}
and this allows us to recover \eqref{PR Hermite} also in this limit. 
\end{rem}

\begin{rem}
Consider the case of purely imaginary $\beta$, i.e. $\beta=i\gamma$,
$\gamma\in\mathbb{R}$. Then the three-term relation \eqref{recur}
generates in the usual way (see e.g. \cite{Deift99}) a symmetric
on $l_{2}$ Jacobi operator, $L^{0}$, defined by the semi-infinite
matrix
\[
L_{n,m}^{0}=R_{n+1}^{1/2}\delta_{n+1,m}+Q_{n}\delta_{n,m}+R_{n}^{1/2}\delta_{n-1,m},\quad n,m\geq0,\;R_{0}=0
\]
whose domain is $D=\left\{ p=\left(p_{0},p_{1},\ldots\right)^{T}\in l_{2}:\;p_{k}=0\mbox{ for sufficiently large }k\right\} $.
Since the moment problem for the measure 
\[
\mathrm{d}\mu\left(x\right)=w\left(x\right)\mathrm{d}x
\]
with the density $w(x)$ given by \eqref{weight} is determinate,
the operator $L^{0}$ is essentially self-adjoint and $\mathrm{d}\mu\left(x\right)$
is the spectral measure of its closure $L\equiv\bar{L^{0}}$. Therefore,
the results of our last two theorems provide an insight into the properties
of semi-infinite Jacobi matrices, i.e. discrete Schrödinger operators
on a half-line, whose spectral densities have discontinuities. In
the earlier works \cite{Johansson98,ChenPrue05,ItsKras08} it was
demonstrated that the discontinuities in the spectral density are
responsible for the oscillatory pattern in the large $n$ asymptotics
of the entries of the Jacobi matrix (in the coordinate asymptotics
of the potentials of the discrete Schrödinger operator). More precisely,
the oscillations occur when the point of the jump of the density is
inside the support of the corresponding equilibrium measure. If the
jump is outside, the behavior of the potentials $R_{n}$ and $Q_{n}$
is monotone. Formulae \eqref{Rnas}, \eqref{Qnas} and \eqref{PR as}
describe the corresponding transition regime. The formulae show that
if the jump happens near the edge of the support then the large n
(coordinate) asymptotics is governed by the Ablowitz-Segur solution
of the second Painlevé equation and the parameters of the solution
are explicitly related to the size of the jump. We actually believe
that this fact is universal, i.e. the transitional formulae will be
the same even if the Gaussian background in the spectral measure is
replaced by an arbitrary exponential weight.
\end{rem}
\medskip{}

Our proofs of Theorem \ref{theorem rec} and Theorem \ref{theorem pn}
are based on the nonlinear steepest descent method of Deift and Zhou
(or, rather on its adaptation \cite{DeKrMcVeZ99a} to the Riemann-Hilbert
(RH) problems related to the orthogonal polynomials \cite{FokItsKit92}).
This method was applied in \cite{XuZhao11} to the case of a discontinuous
Gaussian weight with the point of discontinuity scaled as in \eqref{scaling lambda0}.
We will rely on the transformations and results from this paper, but
we will adapt them in such a way that we can identify the function
$u(t;\kappa)$ as the Painlevé II solution with asymptotics \eqref{u+}
and \eqref{u-}. The RH analysis is presented in Section \ref{section: RH OP},
and the proofs of Theorem \ref{theorem rec} and Theorem \ref{theorem pn}
are given in Section \ref{section: asymptotics recur}.

Theorem \ref{theorem hankel} can be proved in two different ways.
The first one is very short and relies on the Tracy-Widom formula
\eqref{TW} and on known asymptotic results in the Gaussian Unitary
Ensemble. This proof will be given in Section \ref{section: proof hankel TW}.
The second proof, given in Section \ref{section: Hankel}, is lengthy
but has the advantage of being self-contained. It relies on the RH
analysis which we need anyways for the asymptotics of the orthogonal
polynomials and their recurrence coefficients. As is always the case
in the asymptotic analysis of Hankel and Toeplitz determinants, the
move from the asymptotics for the orthogonal polynomials and its recurrence
coefficients to the asymptotics for the Hankel determinants is nontrivial.
One has to address the ``constant of integration problem'' (c.f.
\cite{DeiItsKra11}) which we do with the help of relevant differential
identities for the Hankel determinant $H_{n}(\lambda_{0},\beta)$.

In the RH analysis, we will identify the function $u(t;\kappa)$ as
the solution to the Painlevé II equation with asymptotics (\ref{u+})--(\ref{u-})
using Lax pair arguments and an asymptotic analysis for a certain
model RH problem (see Section \ref{sec: as u+}), which is equivalent
to the one which appeared in \cite{XuZhao11}. Solvability of this
model RH problem was proved in \cite{XuZhao11}, and we prove Theorem
\ref{theorem poles} as a consequence of this in Section \ref{section: poles}.

The analysis in this paper shows similarities with the work \cite{ItsKuiOes09}
where a Painlevé XXXIV function appeared in a parametrix for a different
type of critical edge behavior in unitary random matrix ensembles,
namely with a root singularity instead of the jump singularity which
we consider here. The RH problem which we study differs, however,
from the one analyzed in \cite{ItsKuiOes09}. This yields, in particular,
serious technical differences in the analysis of the large positive
$t$ behavior of the Painlevé transcendent. 
\begin{rem}
As it has already been indicated, it is Painlevé XXXIV equation \eqref{P34}
and the corresponding model RH problem that appear naturally during
the asymptotic analysis of the orthogonal polynomials $p_{n}(x)$.
The solution $y(t;\beta)$ which emerges in this analysis is characterized
by its RH data. We need to transform this characterization into the
asymptotic behavior of $y(t;\beta)$ as $t\to\pm\infty$. Because
of the relation $y=u^{2}$ between the solutions of Painlevé XXXIV
equation \eqref{P34} and the solutions of Painlevé II equation \eqref{P2},
one could think that the needed asymptotics could be extracted from
the work of A. Kapaev \cite{Kapaev92}, where the complete list of
the global asymptotics of the second Painlevé transcendent is presented.
However, to be able to use the results of \cite{Kapaev92} one needs
to connect the RH data of $y(t)$ with the RH data of $u(t)$. A well-known
though still striking fact (see e.g. Chapter 5 of \cite{FokItKaNo06})
is that there is no simple relation between the Lax pair and the RH
problem for the Painlevé XXXIV equation \eqref{P34} and the standard
Flaschka-Newell Lax pair (which is used in \cite{Kapaev92}) and the
RH problem for the Painlevé II equation \eqref{P2}. Hence one does
not know a priori the asymptotics of $u(t)$. There exists, however,
a simple relation between the Lax pair and the RH problem for the
Painlevé XXXIV equation \eqref{P34} and the Lax pair and the RH problem
for the nonuniform second Painlevé equation 
\begin{equation}
q_{tt}=tq+2q^{3}-\frac{1}{2},\label{eq: p12}
\end{equation}
so that one can use \cite{Kapaev92} and determine the asymptotics
of $q(t)$. Unfortunately, now the problem with translation of the
asymptotics of $q(t)$ into the asymptotics of $y(t)$ arises. The
fact of the matter is that the relation between the Painlevé functions
$y(t)$ and $q(t)$ is more complicated than the relation between
the Painlevé functions $y(t)$ and $u(t)$. Indeed, one has that 
\begin{equation}
y(t)=2^{-1/3}U\Bigl(-2^{1/3}t\Bigr),\quad U(t)=q^{2}(t)+q'(t)+\frac{t}{2}\label{eq: qy}
\end{equation}
(see e.g. \cite[Appendix A]{ItsKuiOes09}). This formula virtually
destroys the asymptotic information which one could obtain for the
function $q(t)$ from \cite{Kapaev92}. For instance, one finds from
\cite{Kapaev92} that the function $q(t)$ behaves as $\sim\sqrt{-t/2}$
as $t\to-\infty$. This, as we know \textit{ a posteriori,} must translate
to the exponentially decaying asymptotics of $y(t)$ as $t\to+\infty$.
It is extremely difficult to verify this directly using \eqref{eq: qy}:
one has to prove cancellation of an asymptotic series in all orders
of magnitude. Even worse is the situation with the asymptotics of
$q(t)$ as $t\to+\infty$. It is singular (and is described in terms
of the cotangent function) and, after substitution into \eqref{eq: qy}
it should transform into an oscillatory smooth decaying asymptotics.
We refer the reader to \cite{ItsKuiOes09}, where a similar phenomenon
had already been encountered, for more details. The above discussion
makes it clear that, in spite of the simple relation to the second
Painlevé function $u(t)$, an independent asymptotic analysis of the
Painlevé XXXIV function $y(t)$ is necessary. Of course, it is enough
to evaluate the asymptotics of $y(t;\beta)$ either for $t\to+\infty$
or for $t\to-\infty$, since the one-end asymptotics will enable us
to identify the function $u(x)$ and use \cite{Kapaev92} to determine
its asymptotics on the another end. We have chosen to evaluate the
asymptotics of $y(t;\beta)$ as $t\to+\infty$. The relevant nonlinear
steepest descent analysis is presented in Section \ref{sec: as u+}.
This analysis has some new technical features which are specifically
indicated at the beginning of Section \ref{sec: as u+}.
\end{rem}

\subsection{Applications}

We conclude this introduction by indicating some applications of our
results.

\subsubsection{Random matrix moment generating function}

Consider the $n$-dimensional Gaussian Unitary Ensemble (GUE) normalized
such that the joint eigenvalue probability distribution is given by
\begin{equation}
\frac{1}{Z_{n}}\prod_{1\leq i<j\leq n}(x_{i}-x_{j})^{2}\prod_{j=1}^{n}e^{-x_{j}^{2}}\mathrm{d}x_{j},\qquad x_{1},\ldots,x_{n}\in\mathbb{R}.
\end{equation}
The partition function $Z_{n}$ is then equal to $n!H_{n}\left(\lambda_{0},0\right)$,
with $H_{n}\left(\lambda_{0},0\right)$ given in \eqref{HnGUE}. For
an $n\times n$ GUE matrix, define the random variable $X_{\lambda_{0},n}$
as 
\begin{equation}
X_{\lambda_{0},n}=\mbox{number of eigenvalues greater than \ensuremath{\lambda_{0}}}.
\end{equation}
It is natural to ask how the average of $X_{\lambda_{0},n}$ or its
variance behaves for large $n$. The Hankel determinant with a discontinuous
Gaussian weight carries information about such quantities. Indeed,
the moment generating function of the random variable $X_{\lambda_{0},n}$,
which is defined as $M_{\lambda_{0},n}(y):=\mathbb{E}_{n}\left(e^{yX_{\lambda_{0},n}}\right)$,
can be expressed as 
\begin{equation}
M_{\lambda_{0},n}(y)=\frac{1}{Z_{n}}\int_{\mathbb{R}^{n}}\prod_{1\leq i<j\leq n}(x_{i}-x_{j})^{2}\prod_{j=1}^{n}\left(e^{-x_{j}^{2}}\times\begin{cases}
1, & x_{j}<\lambda_{0}\\
e^{y}, & x_{j}\geq\lambda_{0}
\end{cases}\times\mathrm{d}x_{j}\right).\label{eq: def M}
\end{equation}
This is in fact the ratio of two Hankel determinants, one with a discontinuous
Gaussian weight, and one with a regular Gaussian weight: if we write
$y=-2\pi i\beta$, we have 
\begin{equation}
M_{\lambda_{0},n}(y)=\frac{e^{-\pi in\beta}H_{n}(\lambda_{0},\beta)}{H_{n}\left(\lambda_{0},0\right)}.\label{moment gen}
\end{equation}
This is true for any $n$ and $\lambda_{0}$.

The large $n$ asymptotics for the Hankel determinant $H_{n}$ proved
in Theorem \ref{theorem hankel} together with the explicit expression
\eqref{HnGUE} for $H_{n}\left(\lambda_{0},0\right)$, immediately
give information about the moment generating function as $n\to\infty$
if $\lambda_{0}$ is scaled as in \eqref{scaling lambda0}.

Expanding the moment generating function for small values of $y$,
we have 
\begin{equation}
M_{\lambda_{0},n}(y)=1+y\mathbb{E}_{n}(X_{\lambda_{0},n})+\frac{y^{2}}{2}\mathbb{E}_{n}(X_{\lambda_{0},n}^{2})+\mathcal{O}\left(y^{3}\right),\qquad y\to0,\label{moments}
\end{equation}
so the average and variance of $X_{\lambda_{0},n}$ can be read off
immediately from the small $\beta$ asymptotics for the Hankel determinant.

In particular, differentiating \eqref{hankelas} with respect to $\beta$
and using \eqref{moment gen} and \eqref{moments}, we obtain 
\begin{equation}
\lim_{n\to\infty}\mathbb{E}_{n}(X_{\lambda_{0},n}^{k})=\frac{1}{(-2\pi i)^{k}}\frac{\mathrm{d}^{k}}{\mathrm{d}\beta^{k}}\restr{\left(\exp\left(-\int_{t}^{\infty}(\tau-t)u(\tau;\kappa)^{2}\mathrm{d}\tau\right)\right)}{\beta=0},
\end{equation}
with $\kappa$ given by \eqref{def beta}, which means that the large
$n$ limit of the moments of the random variable $X_{\lambda_{0},n}$
can be expressed in terms of the Ablowitz-Segur Painlevé II solutions
$u(\tau;\kappa)$ and its $\kappa$-derivatives evaluated at $\kappa=0$.
Note that this differentiation is justified since the asymptotics
\eqref{hankelas} are known to be uniform in a small neighborhood
of $\beta=0$. The first $\kappa$-derivative of $u$ is the Airy
function, by \eqref{u kappa0}, and this implies that 
\begin{equation}
\lim_{n\to\infty}\mathbb{E}_{n}(X_{\lambda_{0},n})=\int_{t}^{+\infty}(\tau-t)\Ai(\tau)^{2}\mathrm{d}\tau=\frac{1}{3}\left(2t^{2}\Ai(t)^{2}-\Ai(t)\Ai'(t)-2t\Ai'(t)^{2}\right).
\end{equation}
The same formula can be derived directly from $\lim_{n\to\infty}\mathbb{E}_{n}\left(X_{\lambda_{0},n}\right)=\int_{t}^{+\infty}\rho(\tau)\mathrm{d}\tau$
, where $\rho(t)=K_{\Ai}(t,t)=\Ai'(t)^{2}-\Ai''(t)\Ai(t)$ is the
density for the largest eigenvalue. Similarly, the behavior of higher
moments can also be studied via just the correlation functions $\rho_{m}\left(x_{1},\ldots,x_{m}\right)=\det\left(K_{\Ai}\left(x_{i},x_{j}\right)\right)_{i,j=1}^{m}$.
We would like to thank Peter Forrester for pointing out this fact.

\subsubsection{Largest eigenvalue in a thinned GUE}

The second application is connected to the so-called \textit{thinning
procedure} in the GUE. Consider the $n$ eigenvalues $x_{1}\geq\ldots\geq x_{n}$
of a GUE matrix, and apply the following thinning or filtering procedure
to them: for each eigenvalue independently, we remove it with probability
$s\in(0,1)$. This leads us to a particle configuration, where the
number of remaining particles can be any integer $\ell$ between $0$
and $n$, and we denote those particles by $\mu_{1}\geq\ldots\geq\mu_{\ell}$.
Below, we show that the largest particle distribution in this process
can be expressed in terms of a Hankel determinant with discontinuous
Gaussian weight. More precisely, we have 
\begin{equation}
{\rm Prob}_{s}\left(\mu_{1}\leq\lambda_{0}\right)=M_{\lambda_{0},n}(\log s),\label{prob thinning}
\end{equation}
where $M_{\lambda_{0},n}(t)$ is defined in \eqref{moment gen}.

To prove (\ref{prob thinning}), write $E_{n}(k,\lambda_{0})$ for
the probability that a $n\times n$ GUE matrix has exactly $k$ eigenvalues
bigger than $\lambda_{0}$. If we want none of the thinned or filtered
particles $\mu_{1},\ldots,\mu_{\ell}$ to be bigger than $\lambda_{0}$,
that means that all GUE eigenvalues which are bigger than $\lambda_{0}$
have to be removed by the thinning procedure. Therefore, we have 
\begin{equation}
{\rm Prob}_{s}(\mu_{1}\leq\lambda_{0})=\sum_{k=0}^{n}E_{n}(k,\lambda_{0})s^{k},\label{prob thinning 2}
\end{equation}
since each eigenvalue is removed independently with probability $s$.

Using the integral representation (\ref{eq: def M}), it is on the
other hand straightforward to show that 
\begin{equation}
M_{\lambda_{0},n}(\log s)=\sum_{k=0}^{n}E_{n}(k,\lambda_{0})s^{k}.\label{prob thinning 3}
\end{equation}
Alternatively, this follows from the equation 
\begin{equation}
E_{n}(k,\lambda_{0})=\frac{1}{k!}\left(\frac{d}{ds}\right)^{k}M_{\lambda_{0},n}(\log s),
\end{equation}
which is well-known and proved, for example, in \cite[Ch. 6 and 24]{Mehta04}.
Combining (\ref{prob thinning 2}) with (\ref{prob thinning 3}),
we obtain \eqref{prob thinning}. Consequently, by \eqref{moment gen}
and \eqref{hankelas}, 
\begin{equation}
\lim_{n\to\infty}{\rm Prob}_{s}(\mu_{1}\leq\lambda_{0})=\lim_{n\to\infty}M_{\lambda_{0},n}(\log s)=\exp\left(-\int_{t}^{+\infty}(\tau-t)u(\tau;\kappa)^{2}\mathrm{d}\tau\right),\label{prob thinning 4}
\end{equation}
where $s=1-\kappa^{2}$. This relation, without the Hankel determinant,
was discussed previously in \cite{BohCarPat09,BohigPato04}, where
a transition was observed from the Tracy-Widom distribution (at $s=0$)
to the Weibull distribution (at $s=1$). It is challenging, however,
to describe explicitly the transition asymptotic regime from the behavior
\eqref{det Airy} corresponding to $\beta=-\frac{i}{2\pi}\ln s$,
$0<s\leq1$ to the Tracy-Widom asymptotic behavior, 
\begin{equation}
\ln\det\left(1-\restr{K_{\Ai}}{\left[t,+\infty\right)}\right)=\frac{1}{12}t^{3}-\frac{1}{8}\ln\left(-t\right)+\frac{1}{24}\ln2+\zeta'\left(-1\right)+o\left(1\right),\quad t\to-\infty,
\end{equation}
corresponding to $s=0$, i.e. $\beta=-i\infty$ or $\kappa=1$. Here,
$\zeta$ is the Riemann zeta-function. Similar transition regime for
the sine-kernel determinant has been already described in \cite{BotDeItKr14}
in terms of elliptic functions, and the presence of a very interesting
cascade-type asymptotic behavior has been detected (see also \cite{Dyson95}
where the problem was analyzed, on a heuristic level, for the first
time ). In the case of the Airy-kernel, the question is still open,
although on the level of the logarithmic derivatives, i.e. on the
level of the Painlevé function $u\left(t;\kappa\right)$, the transition
asymptotics from the Ablowitz-Segur case ($\kappa<1$) to the Hastings-McLeod
($\kappa=1$) case has already been found in \cite{Bothner15}.

\subsubsection{Random partitions}

The Airy kernel Fredholm determinant can be interpreted in terms of
random partitions. The Plancherel measure on the set of partitions
of $N\in\mathbb{N}$ is a well-known measure which has its origin
in representation theory. It can be defined in an elementary way by
the following procedure. Take a permutation $\sigma$ in $S_{N}$
and define $x_{1}$ as the maximal length of an increasing subsequence
of $\sigma$. Next, we define $x_{2}$ by requiring that $x_{1}+x_{2}$
is the maximal total length of two disjoint increasing subsequences
of $\sigma$. We proceed in this way, and define $x_{k}$ recursively
by imposing that $x_{1}+\cdots+x_{k}$ is the maximal total length
of $k$ disjoint increasing subsequences of $\sigma$, and we continue
until $x_{1}+\cdots+x_{k}=N$. This procedure associates a partition
$x_{1}\geq\cdots\geq x_{n}$ of $N$ to a permutation $\sigma\in S_{N}$.
The uniform measure on $S_{N}$ induces a measure on the set of partitions
of $N$, which is the Plancherel measure.

We now take a random partition $x_{1}\geq\cdots\geq x_{n}$ of $N$
with respect to the Plancherel measure. Then, the particles $N^{-1/6}(x_{i}-2\sqrt{N})$
converge to the Airy process as $N\to\infty$, see e.g. \cite{Romik15}.
Therefore, if we apply the filtering procedure which removes each
component $x_{i}$ of the partition independently with probability
$s$, we obtain a new partition $\mu_{1}\geq\cdots\geq\mu_{m}$ of
a number $\ell\leq N$. Using similar arguments as in \cite{Romik15},
it follows that 
\begin{equation}
\lim_{N\to\infty}{\rm Prob}_{s}\left(N^{-1/6}(\mu_{1}-2\sqrt{N})\leq t\right)=\det\left(1-\left(1-s\right)\restr{K_{\mathrm{Ai}}}{[t,+\infty)}\right).\label{eq: partitions as}
\end{equation}
Note that the conjectured integral identity \eqref{total integral u}
is of value in relation to \eqref{prob thinning 4} and \eqref{eq: partitions as}.

We want to conclude this section by making the following general remark.
From the point of view of the random matrix theory the examples considered
in this section indicate that, in fact, it is the whole Ablowitz-Segur
family of the Painlevé II transcendents that could appear in the theory
and not only the Hastings-McLeod solution. Regarding the second appearance,
it has already been known due to Bohigas et. al. \cite{BohCarPat09,BohigPato04},
however the first and the third examples are apparently new.

\section{\label{section: proof hankel TW}Theorem \ref{theorem hankel} and
Conjecture \ref{conj Airy}}

\subsection{Proof of Theorem \ref{theorem hankel}}

Denote $K_{n}$ for the GUE eigenvalue correlation kernel 
\begin{equation}
K_{n}(x,y)=e^{-\left(x^{2}+y^{2}\right)/2}\sum_{k=0}^{n-1}H_{k}(x)H_{k}(y),
\end{equation}
built out of normalized degree $k$ Hermite polynomials $H_{k}$,
orthonormal with respect to the weight $e^{-x^{2}}$. Define $G_{\lambda_{0},n}(\kappa)$
by 
\begin{equation}
G_{\lambda_{0},n}(\kappa)=M_{\lambda_{0},n}(\log(1-\kappa^{2}))=\frac{1}{Z_{n}}\intop_{\mathbb{R}^{n}}\prod_{1\leq i<j\leq n}(x_{i}-x_{j})^{2}\prod_{j=1}^{n}\left(e^{-x_{j}^{2}}\times\begin{cases}
1, & x_{j}<\lambda_{0}\\
1-\kappa^{2}, & x_{j}\geq\lambda_{0}
\end{cases}\times\mathrm{d}x_{j}\right).\label{def M}
\end{equation}
By \eqref{moment gen}, we have 
\begin{equation}
G_{\lambda_{0},n}(\kappa)=M_{\lambda_{0},n}(\log(1-\kappa^{2}))=\frac{e^{-\pi in\beta}H_{n}(\lambda_{0},\beta)}{H_{n}\left(\lambda_{0},0\right)},\label{GHankel}
\end{equation}
with $\kappa^{2}=1-e^{-2\pi i\beta}$.

\medskip{}
Similarly as in \eqref{prob thinning 3}, we have 
\begin{equation}
G_{\lambda_{0},n}(\kappa)=M_{\lambda_{0},n}(\log(1-\kappa^{2}))=\sum_{k=0}^{n}(1-\kappa^{2})^{k}E_{n}(k,\lambda_{0})=\det\left(1-\kappa^{2}\restr{K_{n}}{[\lambda_{0},+\infty)}\right),\label{Gdet}
\end{equation}
where $\restr{K_{\Ai}}{[t,+\infty)}$ is the integral operator with
kernel $K_{n}$ acting on $[\lambda_{0},+\infty)$, and the determinant
is the Fredholm determinant (the proof of the last equality in a more
general setting can be found in \cite[§23.3]{Mehta04}).

\medskip{}

Another well-known result is the convergence of the kernel $K_{n}$
to the Airy kernel 
\begin{equation}
K_{\Ai}(x,y)=\frac{\Ai(x)\Ai'(y)-\Ai(y)\Ai'(x)}{x-y}
\end{equation}
if $x,y$ are scaled properly around $\sqrt{2n}$: 
\begin{equation}
\frac{1}{\sqrt{2}n^{1/6}}K_{n}\left(\sqrt{2n}+\frac{u}{\sqrt{2}n^{1/6}},\sqrt{2n}+\frac{v}{\sqrt{2}n^{1/6}}\right)=K_{\Ai}(u,v)+e^{-c(|u|+|v|)}o(1),
\end{equation}
uniformly for $u,v>-M$, $M>0$, for some $c>0$. Using a slightly
stronger version of this Airy kernel limit, as in \cite{DeiftGioe07},
one shows the convergence of the associated Fredholm determinants:
if we scale $\lambda_{0}$ as in \eqref{scaling lambda0}, we have
\begin{equation}
\lim_{n\to\infty}G_{\lambda_{0},n}(\kappa)=\lim_{n\to\infty}\det\left(1-\kappa^{2}\restr{K_{n}}{[\lambda_{0},+\infty)}\right)=\det\left(1-\kappa^{2}\restr{K_{\Ai}}{[t,+\infty)}\right),\label{lim Fredholm}
\end{equation}
uniformly for $t\in(-M,+\infty)$ for any $M>0$, where $\restr{K_{\Ai}}{[t,+\infty)}$
is the integral operator with kernel $K_{\Ai}$ acting on $L^{2}(t,+\infty)$.

Using the Tracy-Widom formula \eqref{TW} together with \eqref{GHankel}
and \eqref{lim Fredholm}, we obtain 
\begin{equation}
H_{n}(\lambda_{0},\beta)=e^{\pi in\beta}H_{n}\left(\lambda_{0},0\right)\exp\left(-\int_{t}^{\infty}(\tau-t)u(\tau;\beta)^{2}\mathrm{d}\tau\right)(1+o(1)),
\end{equation}
as $n\to\infty$. This proves \eqref{hankelas}.

\subsection{Motivation of Conjecture \ref{conj Airy}}

In the case where $\lambda_{0}=\lambda\sqrt{2n}$ with $\lambda\in(-1,1)$,
asymptotics for the Hankel determinants $H_{n}(\lambda_{0},\beta)$
were obtained in \cite{ItsKras08} and are given by \eqref{as Hankel noncrit}.
The dependence of the error term on $\lambda$ was not made explicit
in \cite{ItsKras08}, but it can be seen from their analysis that
the error term in \eqref{as Hankel noncrit} gets worse if $\lambda$
approaches $\pm1$. We hope that by a careful inspection of the estimates
in \cite{ItsKras08}, one can strengthen the error term and obtain
\begin{multline}
H_{n}(\lambda\sqrt{2n},\beta)=H_{n}\left(\lambda_{0},0\right)\,G(1+\beta)G(1-\beta)(1-\lambda^{2})^{-3\beta^{2}/2}(8n)^{-\beta^{2}}\enskip\times\\
\times\enskip\exp\left(2in\beta\left(\arcsin\lambda+\lambda\sqrt{1-\lambda^{2}}\right)\right)\left(1+\bigO\left(\frac{1}{(n^{2/3}(1-\lambda))^{\gamma}}\right)\right),\;\left|\Re\beta\right|<\frac{1}{4},\label{as Hankel noncrit improved}
\end{multline}
for some $\gamma>0$. The error term must be uniform as $\lambda\uparrow1$
at a sufficiently slow rate such that $n^{2/3}(1-\lambda)$ is sufficiently
large, say larger than some fixed $M>0$.

We now take $\lambda=1+tn^{-2/3}/2$ with $-t>2M$. On the one hand,
we can apply \eqref{as Hankel noncrit improved}. Expanding the right-hand
side of \eqref{as Hankel noncrit improved} for large $n$, we obtain,
after a straightforward calculation, 
\begin{multline}
\log H_{n}(\lambda_{0},\beta)-\log H_{n}\left(\lambda_{0},0\right)-\pi in\beta\\
=-\frac{4}{3}i\beta(-t)^{3/2}-\frac{3}{2}\beta^{2}\log(-t)+\log\left(G(1-\beta)G(1+\beta)\right)-3\beta^{2}\log2+\epsilon_{n}(t),\label{Expansion large gap 1}
\end{multline}
where $|\epsilon_{n}(t)|\leq c/|t|^{\gamma}+d\left|t\right|^{5/2}n^{-2/3}$
for some $c,d,\gamma>0$, if $n$ and $-t$ are sufficiently large.
We thank the referees for pointing out the $n$-dependence of this
error term.

On the other hand, by \eqref{hankelas}, we have 
\begin{equation}
\log H_{n}(\lambda_{0},\beta)-\log H_{n}\left(\lambda_{0},0\right)-\pi in\beta=\log\det\left(1-\kappa^{2}\restr{K_{\Ai}}{[t,+\infty)}\right)+o(1),\qquad n\to\infty.\label{Expansion large gap 2}
\end{equation}

Comparing \eqref{Expansion large gap 1} with \eqref{Expansion large gap 2},
we obtain 
\begin{multline}
\log\det\left(1-\kappa^{2}\restr{K_{\Ai}}{[t,+\infty)}\right)\\
=-\frac{4}{3}i\beta(-t)^{3/2}-\frac{3}{2}\beta^{2}\log(-t)+\log\left(G(1-\beta)G(1+\beta)\right)-3\beta^{2}\log2+\epsilon_{n}(t)+o(1),\label{Expansion large gap 3}
\end{multline}
as $n\to\infty$. Letting first $n\to\infty$ and then $t\to-\infty$,
we obtain \eqref{det Airy}. The total integral identity \eqref{total integral u}
now follows easily from \eqref{det Airy} and \eqref{TW}.

\section{RH analysis of orthogonal polynomials\label{section: RH OP} }

\subsection{Overview of transformations\label{sub:Overview-of-transformations}}

Following \cite{FokItsKit92} (see also \cite{Deift99} and \cite{Its11a}),
consider the RH problem for the matrix-valued function $Y(z)$ analytic
in both upper and lower open half-planes with the following jump condition
on the real axis:

\begin{equation}
Y_{+}(x)=Y_{-}(x)\begin{pmatrix}1 & w(x)\\
0 & 1
\end{pmatrix},\;x\in\mathbb{R},\label{eq: jump for Y and weight}
\end{equation}
where $Y_{\pm}(x)$ is the limit of $Y(x)$ as $z$ approaches $x$
from the upper (+) or lower (-) half plane, and with $w(x)$ given
by \eqref{weight}. $Y$ has the asymptotic condition 
\begin{equation}
Y(z)=\left(I+O\left(\frac{1}{z}\right)\right)z^{n\sigma_{3}}\mbox{ as }z\rightarrow\infty,\label{eq: asymptotic condition for Y}
\end{equation}
where $\sigma_{3}$ is the third Pauli matrix 
\[
\sigma_{3}=\begin{pmatrix}1 & 0\\
0 & -1
\end{pmatrix}.
\]

The explicit solution of this problem is 
\begin{equation}
Y(z)=\begin{pmatrix}p_{n}(z) & {\displaystyle \left(2\pi i\right)^{-1}\int_{-\infty}^{\infty}\frac{p_{n}(x)w(x)}{x-z}\mathrm{d}x}\\
-2\pi ih_{n-1}^{-1}p_{n-1}(z) & {\displaystyle -h_{n-1}^{-1}\int_{-\infty}^{\infty}\frac{p_{n-1}(x)w(x)}{x-z}\mathrm{d}x}
\end{pmatrix},\label{eq: explicit solution of the OP RH problem}
\end{equation}
where $p_{n}$ and $p_{n-1}$ are the monic orthogonal polynomials
of degree $n$ and $n-1$ with respect to the weight $w(x)=w(x;\lambda_{0},\beta)$
defined in \eqref{weight}, and $h_{n-1}=\int_{-\infty}^{+\infty}p_{n-1}(x)^{2}w(x)\mathrm{d}x$.

This RH problem for $Y$ has been studied asymptotically, for large
$n$ and with $\lambda_{0}$ scaled as in \eqref{scaling lambda0},
in \cite{XuZhao11}. We give an overview of the series of transformations
constructed in this asymptotic analysis, but refer the reader to \cite{XuZhao11}
for more details. Define the function $T(z)$ as 
\begin{gather}
T(z)=e^{-n\frac{l}{2}\sigma_{3}}\left(2n\right)^{-n\sigma_{3}/2}Y\left(\sqrt{2n}\cdot z\right)e^{n\left(\frac{l}{2}-g(z)\right)\sigma_{3}},\quad z\in\mathbb{C}\setminus\mathbb{R},\label{eq: def of T(z)}
\end{gather}
where 
\begin{equation}
l=-1-2\log2,\quad g(z)=\int_{-1}^{1}\log(z-s)\psi(s)\mathrm{d}s,\;z\in\mathbb{C}\setminus(-\infty,1],\;\psi(s)=\frac{2}{\pi}\sqrt{1-s^{2}}.\label{eq: def of psi(z), l, g(z)}
\end{equation}

Here, the logarithm is in its principle branch with branch cut in
the negative direction, and $\psi(s)>0$ on $\left(-1,1\right)$.
As usual, this $g(z)$ satisfies certain variational relations:
\begin{eqnarray}
g_{+}(z)+g_{-}(z) & = & 2z^{2}+l,\quad z\in\left(-1,1\right),\label{eq: variational property of g 1}\\
g_{+}(z)+g_{-}(z) & < & 2z^{2}+l,\quad z\in\mathbb{R}\setminus\left[-1,1\right].\label{eq: variational property of g 2}
\end{eqnarray}
Additionally, its jump across the real line is described by 
\begin{equation}
g_{+}(z)-g_{-}(z)=\begin{cases}
2\pi i, & z\leq-1,\\
2\pi i\int_{z}^{1}\psi(s)\,\mathrm{d}s, & z\in\left[-1,1\right],\\
0, & z\geq1.
\end{cases}\label{eq: jump of g}
\end{equation}

Let $\psi(z)$ be the analytic continuation of $\psi$ into $\mathbb{C}\setminus\left(\left(-\infty,-1\right]\cup\left[1,\infty\right)\right)$.
Introduce the function $h(z)$ as follows: 
\begin{eqnarray}
h(z) & = & -\pi i\int_{1}^{z}\psi(y)\mathrm{d}y,\;z\in\mathbb{C}\setminus\left(\left(-\infty,-1\right]\cup\left[1,\infty\right)\right),\label{eq: def of h(z)}
\end{eqnarray}
and define a piecewise analytic function $S$ in lens-shaped regions
(see Fig. \ref{fig:The-opening-of lenses}) as follows: 
\begin{equation}
S(z)=T(z)\cdot\begin{cases}
I, & \mbox{outside the lenses},\\
\begin{pmatrix}1 & 0\\
-e^{-i\pi\beta}e^{-2nh(z)} & 1
\end{pmatrix}, & \mbox{in the upper half-lens},\\
\begin{pmatrix}1 & 0\\
e^{-i\pi\beta}e^{2nh(z)} & 1
\end{pmatrix}, & \mbox{in the lower half-lens}.
\end{cases}\label{eq: def of S(z)}
\end{equation}
\begin{figure}
\centering\scalebox{0.8}{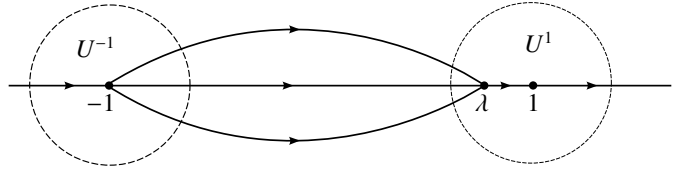}\protect\caption{Opening of lens, case $t<0$.\label{fig:The-opening-of lenses}}
\end{figure}

As shown in \cite{XuZhao11}, the function $S$ has jumps on the lens-shaped
contour shown in Fig. \ref{fig:The-opening-of lenses}. As $n\to\infty$,
the jump matrices tend to the identity matrix everywhere except on
$(-1,1)$ and in small disks $U^{-1}$ and $U^{1}$ around $-1$ and
$1$. To obtain asymptotics for $S$, an outer parametrix and local
parametrices near $-1$ and $+1$ have to be constructed.

\subsection{Outer parametrix}

For $z$ outside small disks around $-1$ and $+1$, $S$ can be approximated
for large $n$ by an outer parametrix $P^{(\infty)}$, which is analytic
except on $[-1,1]$, tends to the identity as $z\to\infty$, and has
the jump relation 
\begin{equation}
P_{+}^{(\infty)}(z)=P_{-}^{(\infty)}(z)\cdot\begin{pmatrix}0 & e^{\pi i\beta}\\
-e^{-\pi i\beta} & 0
\end{pmatrix}.\quad z\in(-1,1).\label{eq: jumps of P(infty)}
\end{equation}
It is given explicitly (see e.g. \cite{ItsKras08}) as 
\begin{equation}
P^{(\infty)}(z)=\frac{1}{2}e^{i\pi\beta\sigma_{3}/2}\begin{pmatrix}a_{0}+a_{0}^{-1} & -i\left(a_{0}-a_{0}^{-1}\right)\\
i\left(a_{0}-a_{0}^{-1}\right) & a_{0}+a_{0}^{-1}
\end{pmatrix}e^{-i\pi\beta\sigma_{3}/2},\label{eq: definition of P(infty)}
\end{equation}
where 
\begin{equation}
a_{0}(z)=\left(\frac{z-1}{z+1}\right)^{1/4}.\label{eq: def of a_0, a, b}
\end{equation}
The branch of $a_{0}$ is chosen so that $a_{0}(z)\to1$ as $z\to\infty$.

\subsection{Local parametrix near 1\label{sub:Local-parametrix-near-1}}

In order to obtain asymptotics for $S$ also in neighborhoods of $\pm1$,
local parametrices have been constructed in \cite{XuZhao11}. Near
$-1$, this local parametrix was built using the Airy function, but
we do not need its explicit form. Near $+1$, it was built using a
model RH problem associated to the Painlevé XXXIV equation.

The local parametrix $P^{(1)}(z)$ is analytic in $U^{1}$, except
for $z$ on the jump contour for $S$, and it has the same jump relations
as $S$ for $z$ on the jump contour for $S$, inside $U^{1}$. On
the boundary $\partial U^{1}$, it satisfies the matching condition
\begin{equation}
P^{(1)}(z)\cdot P^{(\infty)}(z)^{-1}=I+\bigO(n^{-1/3})\mbox{ as }n\rightarrow\infty,\mbox{ uniformly for }z\in\partial U^{1}.\label{eq:gluing condition}
\end{equation}
It takes the form 
\begin{equation}
P^{(1)}(z)=E(z)\Phi(\zeta(z);\tau)e^{\frac{2}{3}\zeta(z)^{3/2}\sigma_{3}}e^{-i\pi\beta\sigma_{3}/2},\label{eq: P(1) parametrix}
\end{equation}
where $E$ is an analytic function in $U^{1}$, $\Phi$ will be specified
below, and $\zeta(z)$ is a conformal map near $1$. The conformal
map $\zeta(z)$ and the parameter $\tau$ are given by 
\begin{equation}
\zeta(z)=\left(-\frac{3}{2}nh(z)\right)^{2/3},\qquad\tau=\zeta(\lambda)=\zeta\left(1+\frac{t}{2}n^{-2/3}\right).\label{eq: definition of zeta}
\end{equation}
Here, the multivalued power is the principle branch in $\mathbb{C}\setminus\left(-\infty,0\right)$.
The Taylor expansion at $1$ is 
\begin{eqnarray}
\zeta(z) & = & 2n^{2/3}\left(z-1\right)\left(1+\frac{1}{10}\left(z-1\right)+\mathcal{O}\left(z-1\right)^{2}\right)\mbox{ as }z\rightarrow1,\label{eq: expansion zeta}\\
\tau & = & t+\mathcal{O}\left(n^{-2/3}\right)\mbox{ as }n\to\infty.\label{eq: expansion tau}
\end{eqnarray}

The analytic pre-factor $E$ can be expressed as 
\begin{equation}
E(z)=P^{(\infty)}(z)e^{i\pi\beta\sigma_{3}/2}\frac{1}{\sqrt{2}}\begin{pmatrix}1 & i\\
i & 1
\end{pmatrix}\zeta(z)^{-\sigma_{3}/4},\label{eq: definition of E(z)}
\end{equation}
and $\Phi(\zeta;\tau)$ is given by 
\begin{equation}
\Psi_{0}(\zeta;\tau)=\begin{pmatrix}1 & \frac{i\tau^{2}}{4}\\
0 & 1
\end{pmatrix}\Phi(\zeta+\tau;\tau),\label{def Phi}
\end{equation}
where $\Psi_{0}(\xi;\tau)$ is the solution to the following RH problem.

\begin{figure}[h]
\centering{}\scalebox{0.8}{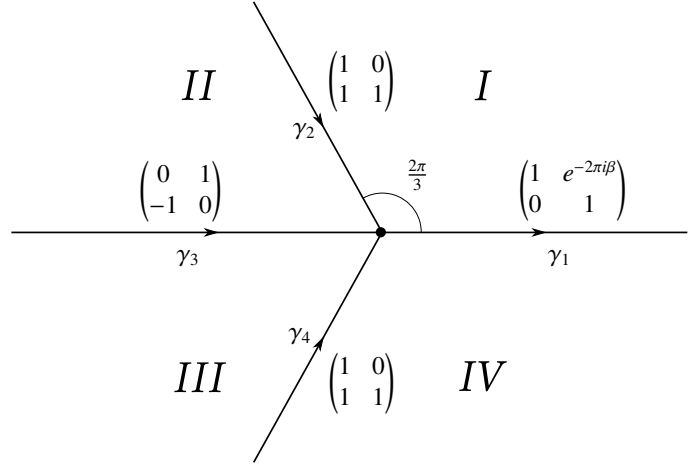}\protect\protect\caption{The RH problem for $\Psi_{0}(\xi)$. The rays meet at $\xi=0$. The
union of the rays is referred to as $\Gamma_{\Psi_{0}}$.\label{fig: Model problem contours}}
\end{figure}

$\Psi_{0}$ is analytic off the contour shown in Fig.\ \ref{fig: Model problem contours}
and satisfies the following jump and asymptotic conditions: 
\begin{equation}
\Psi_{0+}(\xi)=\Psi_{0-}(\xi)\cdot\begin{cases}
\begin{pmatrix}1 & e^{-2\pi i\beta}\\
0 & 1
\end{pmatrix}, & \xi\in\gamma_{1},\\
\begin{pmatrix}1 & 0\\
1 & 1
\end{pmatrix}, & \xi\in\gamma_{2}\cup\gamma_{4},\\
\begin{pmatrix}0 & 1\\
-1 & 0
\end{pmatrix}, & \xi\in\gamma_{3}.
\end{cases}\label{eq:model problem jumps-1}
\end{equation}
\begin{eqnarray}
\Psi_{0}(\xi) & = & \left(I+\frac{m}{\xi}+\mathcal{O}\left(\frac{1}{\xi^{2}}\right)\right)\xi^{\sigma_{3}/4}\frac{1}{\sqrt{2}}\begin{pmatrix}1 & -i\\
-i & 1
\end{pmatrix}e^{-\left(\frac{2}{3}\xi^{3/2}+\tau\xi^{1/2}\right)\sigma_{3}}\mbox{ as }\xi\rightarrow\infty,\label{eq:model asymptotic at infinity}\\
\Psi_{0}(\xi) & = & \left(\begin{pmatrix}a & b\\
c & d
\end{pmatrix}+\mathcal{O}(\xi)\right)\left(I+\frac{\kappa^{2}}{2\pi i}\begin{pmatrix}0 & 1\\
0 & 0
\end{pmatrix}\log\xi\right)M(\xi)\mbox{ as }\xi\rightarrow0,\label{eq: model behavior at zero}
\end{eqnarray}
where the matrix elements of $m$ as well as $a$, $b$, $c$, $d$
are some functions of $\tau$, $\kappa$ is given by \eqref{def beta},
and $M$ is a piecewise constant function defined as follows 
\begin{equation}
M(\xi)=\begin{cases}
\begin{pmatrix}1 & 0\\
0 & 1
\end{pmatrix}, & \xi\in I,\\
\begin{pmatrix}1 & 0\\
-1 & 1
\end{pmatrix}, & \xi\in II,\\
\begin{pmatrix}1-e^{-2\pi i\beta} & -e^{-2\pi i\beta}\\
1 & 1
\end{pmatrix}, & \xi\in III,\\
\begin{pmatrix}1 & -e^{-2\pi i\beta}\\
0 & 1
\end{pmatrix}, & \xi\in IV.
\end{cases}\label{eq:constant factors M}
\end{equation}
All multivalued functions above are in their principle branches with
branch cuts along the negative half axis.$\Psi_{0}$ is uniquely determined
by the above conditions. Note that all higher order terms in the expansions
of $\Psi_{0}$ are also functions of $\tau$.

The function $P^{(1)}$ defined in \eqref{eq: P(1) parametrix} is
the same as the one in \cite{XuZhao11}, but it has to be noted that
our functions $\Psi_{0}$ is defined in a slightly different way compared
to \cite{XuZhao11}, which will be convenient later on. We have the
relation 
\begin{equation}
\Psi_{0}^{XZ}(\zeta;s)=\begin{pmatrix}0 & i\\
i & 0
\end{pmatrix}2^{-\sigma_{3}/6}\Psi_{0}\left(\xi=2^{2/3}\zeta;\tau=-2^{-1/3}s\right),\label{relation Psi0}
\end{equation}
where $\Psi_{0}^{XZ}$ denotes the solution to the model RH problem
of \cite{XuZhao11}.

By \eqref{def Phi} and \eqref{eq:model asymptotic at infinity},
it is straightforward to verify that $\Phi$ admits the asymptotic
expansion

\begin{equation}
\Phi(\zeta;\tau)=\left(I+\frac{m^{\Phi}}{\zeta}+\mathcal{O}\left(\frac{1}{\zeta^{2}}\right)\right)\zeta^{\sigma_{3}/4}\frac{1}{\sqrt{2}}\begin{pmatrix}1 & -i\\
-i & 1
\end{pmatrix}e^{-\frac{2}{3}\zeta^{3/2}\sigma_{3}}\mbox{ as }\zeta\rightarrow\infty,\label{eq: asymptotic condition for Phi(z)}
\end{equation}
where we have the following relation between $m=m(\tau)$ and $m^{\Phi}=m^{\Phi}(\tau)$,
\begin{equation}
m=\begin{pmatrix}{\displaystyle m_{11}^{\Phi}+\frac{i\tau^{2}}{4}m_{21}^{\Phi}+\frac{\tau}{4}-\frac{\tau^{4}}{32}} & {\displaystyle m_{12}^{\Phi}-\frac{i\tau^{2}}{4}m_{11}^{\Phi}+\frac{\tau^{4}}{16}m_{21}^{\Phi}-\frac{i\tau^{3}}{12}+\frac{i\tau^{6}}{192}}\\
{\displaystyle m_{21}^{\Phi}+\frac{i\tau^{2}}{4}} & {\displaystyle m_{22}^{\Phi}-\frac{i\tau^{2}}{4}m_{21}^{\Phi}-\frac{\tau}{4}+\frac{\tau^{4}}{32}}
\end{pmatrix}.\label{eq: m in terms of mPhi}
\end{equation}

\subsection[Lax pair for \texorpdfstring{$\Psi_{0}$}{Psi0} and the Painlevé
XXXIV equation]{\label{section: Lax}Lax pair for \boldmath{$\Psi_{0}$} and
the Painlevé XXXIV equation }

From the RH conditions for $\Psi_{0}$, there is a standard procedure
to deduce differential equations with respect to the variable $\xi$
and the parameter $\tau$. Here, our approach deviates from the one
in \cite{XuZhao11}.

Consider the functions $U:=\frac{\partial\Psi_{0}}{\partial\xi}\Psi_{0}^{-1}$
and $V=\frac{\partial\Psi_{0}}{\partial\tau}\Psi_{0}^{-1}$. Because
the jump matrices for $\Psi_{0}$ are independent of $\xi$ and $\tau$,
$U$ and $V$ are meromorphic functions of $\xi$. Using the behavior
of $\Psi_{0}$ at infinity and $0$ given in \eqref{eq:model asymptotic at infinity}
and \eqref{eq: model behavior at zero}, we obtain after a straightforward
calculation that $\Psi_{0}$ satisfies the Lax pair 
\begin{eqnarray}
\frac{\partial\Psi_{0}}{\partial\xi}(\xi;\tau) & = & U(\xi;\tau)\Psi_{0}(\xi;\tau),\quad U(\xi;\tau)=V(\xi;\tau)+\begin{pmatrix}0 & -i\tau/2\\
0 & 0
\end{pmatrix}+\frac{\kappa^{2}}{2\pi i\xi}\begin{pmatrix}-ac & a^{2}\\
-c^{2} & ac
\end{pmatrix},\label{eq:Lax system xi}\\
\frac{\partial\Psi_{0}}{\partial\tau}(\xi;\tau) & = & V(\xi;\tau)\Psi_{0}(\xi;\tau),\quad V(\xi;\tau)=-i\xi\begin{pmatrix}0 & 1\\
0 & 0
\end{pmatrix}-i\begin{pmatrix}-m_{21} & 2m_{11}\\
-1 & m_{21}
\end{pmatrix},\label{eq:Lax system tau}
\end{eqnarray}
where $a,b,c,d$ and the matrix $m$, which are functions of the parameter
$\tau$ (and also of $\beta$), were defined in \eqref{eq:model asymptotic at infinity}-\eqref{eq: model behavior at zero}.
We can also refine the expansion for $\frac{\partial\Psi_{0}}{\partial\tau}\Psi_{0}^{-1}$
as $\xi\rightarrow\infty$: 
\begin{equation}
\frac{\partial\Psi_{0}}{\partial\tau}\Psi_{0}^{-1}-V(\xi,\tau)=\frac{1}{\xi}\frac{\mathrm{d}m}{\mathrm{d}\tau}-\frac{i}{\xi}\left[m,\begin{pmatrix}0 & 0\\
-1 & 0
\end{pmatrix}\right]-\frac{i}{\xi}\left[m,\begin{pmatrix}0 & 1\\
0 & 0
\end{pmatrix}\right]-\frac{i}{\xi}\left[\begin{pmatrix}0 & 1\\
0 & 0
\end{pmatrix}m,m\right]+\mathcal{O}\left(\frac{1}{\xi^{2}}\right).\label{eq: Lax additional equation}
\end{equation}
Since this expression obviously has to be zero, equating its $(21)$
entry to zero gives us the useful relation 
\begin{equation}
m_{11}=\frac{1}{2}m_{21}^{2}-\frac{i}{2}m'_{21}.\label{eq:m11 through m21}
\end{equation}

Note that $m_{21}^{2}$ is the square of the matrix element $m_{21}$.
The compatibility condition of the Lax system (\ref{eq:Lax system xi})-(\ref{eq:Lax system tau}),
\begin{equation}
V_{\xi}-U_{\tau}=\left[U,V\right],\label{eq: Lax compatibility condition}
\end{equation}
becomes 
\begin{multline}
-i\begin{pmatrix}0 & 1\\
0 & 0
\end{pmatrix}+i\begin{pmatrix}-m_{21}' & 2m_{11}'\\
0 & m_{21}'
\end{pmatrix}-\frac{\kappa^{2}}{2\pi i\xi}\begin{pmatrix}-\left(ac\right)' & \left(a^{2}\right)'\\
-\left(c^{2}\right)' & \left(ac\right)'
\end{pmatrix}=\\
=\begin{pmatrix}\frac{\tau}{2} & -\tau m_{21}\\
0 & -\frac{\tau}{2}
\end{pmatrix}-i\frac{\kappa^{2}}{2\pi i}\begin{pmatrix}c^{2} & -2ac\\
0 & -c^{2}
\end{pmatrix}-\frac{i\kappa^{2}}{2\pi i\xi}\begin{pmatrix}2c^{2}m_{11}-a^{2} & 2a^{2}m_{21}-4acm_{11}\\
2c^{2}m_{21}-2ac & a^{2}-2c^{2}m_{11}
\end{pmatrix}.\label{eq: Lax compatibility equation}
\end{multline}
This equation can be separated into two equations for each power of
$\xi$. From the resulting system one can extract the equations 
\begin{eqnarray}
\frac{\kappa^{2}}{2\pi i}c^{2}=m_{21}'(\tau)-\frac{i\tau}{2} & = & \left(m_{21}^{\Phi}\right)',\label{eq: connection c^2 to m_21}\\
\frac{\kappa^{2}}{2\pi i}\left(ac-\frac{i\tau^{2}}{4}c^{2}\right) & = & \left(m_{11}^{\Phi}\right)',\label{eq: (ac-itau^2 c^2) is m11Phi'}
\end{eqnarray}
and 
\begin{equation}
\left(1+2i\tau m_{21}-4m_{11}'\right)^{2}+4\left(2m_{21}'-i\tau\right)\left(2im_{11}''+2im_{11}\left(\tau+2im_{21}'\right)+\tau m_{21}'+m_{21}\right)=0,\label{eq: Lax differential equation mixed}
\end{equation}
which, with the help of \eqref{eq:m11 through m21}, reduces to 
\begin{equation}
1+32\tau\left(m_{21}'\right)^{2}+32i\left(m_{21}'\right)^{3}+4im_{21}''-4\left(m_{21}''\right)^{2}-4i\tau m_{21}'''+8m_{21}'\left(m_{21}'''-i\tau^{2}\right)=0.\label{eq: Lax differential equation for m_21}
\end{equation}
This equation is a disguised version of the 34th Painlevé equation
for the function 
\begin{equation}
y(\tau)=-im_{21}'(\tau)-\frac{\tau}{2}=-i\left(m_{21}^{\Phi}\right)'(\tau),\label{eq: def of y(tau)}
\end{equation}
namely, 
\begin{equation}
y_{\tau\tau}=4y^{2}+2\tau y+\frac{(y_{\tau})^{2}}{2y}.\label{eq: 34th painleve for y}
\end{equation}

Equation \eqref{eq: connection c^2 to m_21} also provides us with
another representation of $y(\tau)$: 
\begin{equation}
y(\tau)=i\lim_{\xi\rightarrow0}\left[\xi\frac{\mathrm{d}\Psi_{0}(\xi)}{\mathrm{d}\xi}\Psi_{0}^{-1}(\xi)\right]_{21}.\label{eq:Limit definition of y(tau)}
\end{equation}
Moreover, from \eqref{eq:m11 through m21}, we obtain an additional
expression for $y$ which does not involve derivatives: 
\begin{equation}
y(\tau)=2m_{11}(\tau)-m_{21}^{2}(\tau)-\tau/2.\label{eq: y in terms of m or mPhi without derivatives}
\end{equation}

\subsection{Proof of Theorem \ref{theorem poles}}

\label{section: poles}

In \cite[Corollary 1]{XuZhao11}, it was proved using vanishing lemma
techniques that the RH problem for $\Psi_{0}^{XZ}(\zeta;s)$ is solvable
for all real values of $s$ if $\beta$ is such that $\left|\arg e^{-2i\pi\beta}\right|<\pi$,
and thus for all $\beta$ such that $|\Re\beta|<1/2$. Because of
the explicit relation \eqref{relation Psi0}, this implies that the
RH problem for $\Psi_{0}$ is also solvable for all real values of
$\tau$ if $|\Re\beta|<1/2$. This in turn implies that the function
$y(\tau)=y(\tau;\beta)$ defined in terms of $\Psi_{0}$ by \eqref{eq: y in terms of m or mPhi without derivatives}
is well-defined and cannot have singularities for real $\tau$ if
$|\Re\beta|<1/2$.

If we define $u(\tau;\kappa)$ by $u(\tau;\kappa)^{2}=y(\tau;\beta)$
with $\kappa^{2}=1-e^{-2\pi i\beta}$, then it is easily verified
by \eqref{eq: 34th painleve for y} that $u$ solves the Painlevé
II equation \eqref{P2}. By exploring the asymptotic behavior of $y(\tau;\beta)$
(or, equivalently, of $u(\tau;\kappa)$) as $\tau\rightarrow\pm\infty$,
we will be able to identify $u(\tau;\kappa)$ as the Ablowitz-Segur
solution of the Painlevé II equation characterized by \eqref{u+}
and \eqref{u-}. This identification, which will follow from \eqref{eq: asymptotic for y(tau) at +Infty}
below, completes the proof of Theorem \ref{theorem poles}.

\subsection{Final transformation}

Introduce the new function 
\begin{equation}
R(z)=S(z)\cdot\begin{cases}
\left(P^{(\infty)}(z)\right)^{-1}, & z\in\mathbb{C}\setminus\overline{U^{-1}\cup U^{1}\cup(-1,1)},\\
\left(P^{(-1)}(z)\right)^{-1}, & z\in U^{(-1)},\\
\left(P^{(1)}(z)\right)^{-1}, & z\in U^{(1)},
\end{cases}
\end{equation}
which tends to the identity matrix as $z\to\infty$ and which has
jump matrices $G_{R}$ on the contour $\Gamma_{R}$ that tend to identity
as $n\to\infty$: 
\begin{equation}
G_{R}(z)=\begin{cases}
P_{+}^{(\infty)}(z)\left(P_{-}^{(\infty)}(z)\right)^{-1}, & z\in\left(-1,1\right)\setminus\overline{U^{-1}\cup U^{1}},\\
P^{(\infty)}(z)\left(P^{(-1)}(z)\right)^{-1}, & z\in\partial U^{-1},\\
P^{(\infty)}(z)\left(P^{(1)}(z)\right)^{-1}, & z\in\partial U^{1},
\end{cases}=I+\bigO\left(\frac{1}{n^{1/3}(1+|z|)^{p}}\right),\;z\in\Gamma_{R}.\label{eq: G_R is everywhere close to I}
\end{equation}
This, in turn, implies (see \cite{DeiftZhou93}) that for sufficiently
large $n$ 
\begin{equation}
R(z)=I+\mathcal{O}\left(\frac{1}{(1+\left|z\right|)n^{1/3}}\right),\mbox{ uniformly for }z\in\mathbb{C}\setminus\Gamma_{R},\label{asymptotics R}
\end{equation}
where $\Gamma_{R}$ is the jump contour for $R$.

\section[Asymptotics of \texorpdfstring{$u(\tau;\kappa)$ as $\tau\rightarrow+\infty$}{u(t;k)
as t->+infinity}]{\label{sec: as u+}Asymptotics of \boldmath{$u(\tau;\kappa)$
as $\tau\rightarrow+\infty$}}

From Section \ref{section: Lax}, we know that $y(\tau;\beta)$ defined
by \eqref{eq: def of y(tau)} solves the Painlevé XXXIV equation \eqref{eq: 34th painleve for y},
and this implies that $u$ defined by $u(\tau;\kappa)^{2}=y(\tau;\beta)$
(with relation \eqref{def beta} between $\kappa$ and $\beta$) solves
the Painlevé II equation \eqref{P2}. We now proceed with proving
the asymptotics of $y(\tau;\beta)=u(\tau;\kappa)^{2}$ as stated in
\eqref{u+}. In this section it is assumed that $\tau>0$. The analysis
performed here is largely analogous to the one in \cite{ItsKuiOes09}
with one additional new technical feature -- the need to introduce
an additional triangular parametrix near the point $z=0$ of the discontinuity
of the (triangular) jump matrix (see Subsection \ref{sub: 4.3.2}).

\subsection{Rescaling and shift of the jump contour}

Introduce 
\begin{equation}
A(z)=\tau^{-\sigma_{3}/4}\Psi_{0}(\tau z;\tau).\label{eq: def of A+}
\end{equation}
One can easily check that it satisfies the following RH problem. 
\begin{itemize}
\item[(a)] $A:\mathbb{C}\setminus\Gamma_{\Psi_{0}}\rightarrow\mathbb{C}^{2\times2}$
is analytic. 
\item[(b)] $A$ has the same jump relations as $\Psi_{0}$. 
\item[(c)] $A(z)=\left(I+O\left(\frac{1}{z}\right)\right)z^{\sigma_{3}/4}\frac{1}{\sqrt{2}}\begin{pmatrix}1 & -i\\
-i & 1
\end{pmatrix}e^{-\tau^{3/2}\left(\frac{2}{3}z^{3/2}+z^{1/2}\right)\sigma_{3}}$ as $z\rightarrow\infty$. 
\item[(d)] $A(z)$ inherits its behavior at $z=0$ from $\Psi_{0}$ very easily.
\end{itemize}
From (\ref{eq:Limit definition of y(tau)}) we get 
\begin{equation}
y(\tau)=\frac{i}{\sqrt{\tau}}\lim_{z\rightarrow0}\left[z\frac{\mathrm{d}A(z)}{\mathrm{d}z}A(z)^{-1}\right]_{21}.\label{eq: y(tau) in terms of A+}
\end{equation}

We shall further write 
\begin{equation}
s=\tau^{3/2}.\label{eq:definition of s+}
\end{equation}
With respect to the domains defined in Fig.\ \ref{fig: RH for B},
define 
\begin{equation}
B(z)=\begin{cases}
A(z)\cdot\begin{pmatrix}1 & 0\\
1 & 1
\end{pmatrix}, & z\in II',\\
A(z)\cdot\begin{pmatrix}1 & 0\\
-1 & 1
\end{pmatrix}, & z\in III',\\
A(z), & z\in I\cup II''\cup III''\cup IV.
\end{cases}\label{eq: def of B+}
\end{equation}
\begin{figure}[h]
\centering{}\scalebox{0.8}{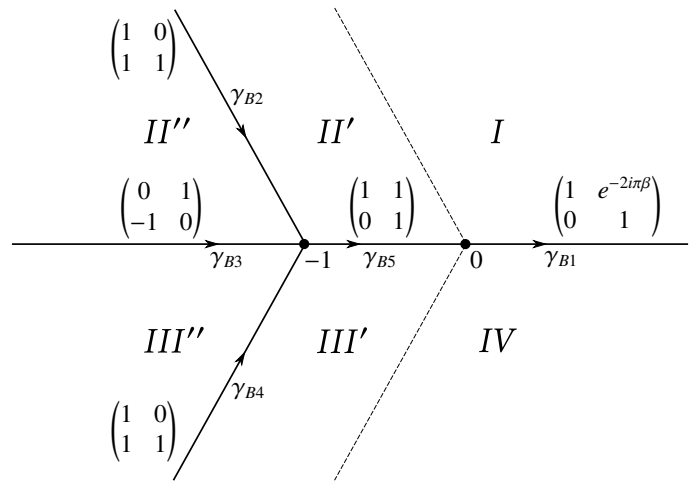}\protect\caption{The contours $\Gamma_{B}$ and the RH problem for $B(z)$.\label{fig: RH for B}}
\end{figure}

This function satisfies the following RH problem. 
\begin{itemize}
\item[(a)] $B:\mathbb{C}\setminus\Gamma_{B}\rightarrow\mathbb{C}^{2\times2}$
is analytic. 
\item[(b)] $B_{+}(z)=B_{-}(z)\cdot\begin{cases}
\begin{pmatrix}1 & e^{-2\pi i\beta}\\
0 & 1
\end{pmatrix}, & z\in\gamma_{B1},\\
\begin{pmatrix}1 & 0\\
1 & 1
\end{pmatrix}, & z\in\gamma_{B2}\cup\gamma_{B4},\\
\begin{pmatrix}0 & 1\\
-1 & 0
\end{pmatrix}, & z\in\gamma_{B3},\\
\begin{pmatrix}1 & 1\\
0 & 1
\end{pmatrix}, & z\in\gamma_{B5}.
\end{cases}$ 
\item[(c)] $B(z)=\left(I+O\left(\frac{1}{z}\right)\right)z^{\sigma_{3}/4}\frac{1}{\sqrt{2}}\begin{pmatrix}1 & -i\\
-i & 1
\end{pmatrix}e^{-s\left(\frac{2}{3}z^{3/2}+z^{1/2}\right)\sigma_{3}}$ as $z\to\infty$. 
\item[(d)] $B(z)$ has logarithmic behavior near $z=0$. Namely, 
\begin{equation}
B(z)=\widetilde{B}(z)\left(I+\frac{\kappa^{2}}{2\pi i}\begin{pmatrix}0 & 1\\
0 & 0
\end{pmatrix}\log z\right)M_{\pm},\quad z\in\mathbb{H}^{\pm},\label{eq: behavior at zero of B+}
\end{equation}
where $\widetilde{B}(z)$ is some analytic function, $M_{+}=I$ and
$M_{-}=\begin{pmatrix}1 & -e^{-2\pi i\beta}\\
0 & 1
\end{pmatrix}$. The logarithm is in its principle branch with branch cut along the
negative half axis.
\end{itemize}
Finally, the expression for $y(\tau)$ remains unchanged compared
to \eqref{eq: y(tau) in terms of A+}, 
\begin{equation}
y(\tau)=\frac{i}{\sqrt{\tau}}\lim_{z\rightarrow0}\left[z\frac{\mathrm{d}B(z)}{\mathrm{d}z}B(z)^{-1}\right]_{21}.\label{eq: y(tau) in terms of B+}
\end{equation}

This transformation is an important precursor of the introduction
of a new $g$-function that will allow us to ``undress'' the behavior
of the RH problem at infinity (for similar transitions see e.g. \cite{ItsKuiOes09}
and \cite{XuDaiZha14}).

\subsection{Normalization at infinity}

We now introduce the following $g$-function, 
\begin{equation}
\widehat{g}(z)=\frac{2}{3}(z+1)^{3/2},\quad-\pi<\arg(z+1)<\pi.\label{eq:g_m-function}
\end{equation}
Note that 
\begin{equation}
\widehat{g}(z)-\left(\frac{2}{3}z^{3/2}+z^{1/2}\right)=\frac{1}{4z^{1/2}}+\mathcal{O}\left(\frac{1}{z^{3/2}}\right)\mbox{ as }z\rightarrow\infty,\;z\notin(-\infty,-1].\label{eq: expansion of g_m(z)}
\end{equation}
Next, define 
\begin{equation}
C(z)=\begin{pmatrix}1 & -is/4\\
0 & 1
\end{pmatrix}B(z)e^{s\widehat{g}(z)\sigma_{3}}.\label{eq: def of C+}
\end{equation}
The constant prefactor in this definition is needed to conserve the
$O\left(\frac{1}{z}\right)$ term in the asymptotics as $z\rightarrow\infty$.
$C$ satisfies the following RH problem. 
\begin{itemize}
\item[(a)] $C:\mathbb{C}\setminus\Gamma_{B}\rightarrow\mathbb{C}^{2\times2}$
is analytic. 
\item[(b)] $C_{+}(z)=C_{-}(z)\cdot\begin{cases}
\begin{pmatrix}1 & e^{-2\pi i\beta}e^{-2s\widehat{g}(z)}\\
0 & 1
\end{pmatrix}, & z\in\gamma_{B1},\\
\begin{pmatrix}1 & 0\\
e^{2s\widehat{g}(z)} & 1
\end{pmatrix}, & z\in\gamma_{B2}\cup\gamma_{B4},\\
\begin{pmatrix}0 & 1\\
-1 & 0
\end{pmatrix}, & z\in\gamma_{B3},\\
\begin{pmatrix}1 & e^{-2s\widehat{g}(z)}\\
0 & 1
\end{pmatrix}, & z\in\gamma_{B5}.
\end{cases}$ 
\item[(c)] $C(z)=\left(I+\mathcal{O}\left(\frac{1}{z}\right)\right)z^{\sigma_{3}/4}\frac{1}{\sqrt{2}}\begin{pmatrix}1 & -i\\
-i & 1
\end{pmatrix}$ as $z\rightarrow\infty$, $z\notin\Gamma_{B}$. 
\item[(d)] $C(z)=\widetilde{C}(z)\left(I+\frac{\kappa^{2}}{2\pi i}\begin{pmatrix}0 & 1\\
0 & 0
\end{pmatrix}\log z\right)M_{\pm}e^{s\widehat{g}(z)\sigma_{3}},$ for $z\in\mathbb{C}^{\pm}$ near 0, where $\widetilde{C}$ is analytic
in a neighborhood of $0$ and $M_{\pm}$ are the same as before. The
branch of the logarithm is chosen as before.
\end{itemize}
From the definition of $C$, 
\begin{equation}
\left[z\frac{\mathrm{d}C(z)}{\mathrm{d}z}C(z)^{-1}\right]_{21}=\left[z\frac{\mathrm{d}B(z)}{\mathrm{d}z}B(z)^{-1}\right]_{21}+s\widehat{g}'(z)\left[zB(z)\sigma_{3}B(z)^{-1}\right]_{21}.\label{eq: useful equation for C+}
\end{equation}
The second term in this expression tends to zero as $z\rightarrow0$
due to the behavior of $B(z)$, hence 
\begin{equation}
y(\tau)=\frac{i}{\sqrt{\tau}}\lim_{z\rightarrow0}\left[z\frac{\mathrm{d}C(z)}{\mathrm{d}z}C(z)^{-1}\right]_{21}.\label{eq: y(tau) in terms of C(z)}
\end{equation}

\subsection{Construction of parametrices}

\subsubsection[Global Airy solution \texorpdfstring{$C^{(\Ai)}$}{C(Ai)}]{Global Airy solution \boldmath{$C^{(\Ai)}$}}

The jumps of $C(z)$ near $z=-1$ are very similar to the jumps of
the standard Airy RH problem. Let us look for a function $C^{(\Ai)}$
that satisfies the following RH problem. 
\begin{itemize}
\item[(a)] $C^{(\Ai)}:\mathbb{C}\setminus\Gamma_{B}\rightarrow\mathbb{C}^{2\times2}$
is analytic. 
\item[(b)] $C^{(\Ai)}(z)$ has the same jumps on $\Gamma_{B}\setminus[-1,+\infty)$
as $C(z)$ and its jump on $(-1,+\infty)$ is $\begin{pmatrix}1 & e^{-2s\widehat{g}(z)}\\
0 & 1
\end{pmatrix}$. 
\item[(c)] $C^{(\Ai)}(z)=\left(I+\mathcal{O}\left(\frac{1}{z}\right)\right)z^{\sigma_{3}/4}\frac{1}{\sqrt{2}}\begin{pmatrix}1 & -i\\
-i & 1
\end{pmatrix}$ as $z\rightarrow\infty$. 
\end{itemize}
We seek $C^{(\Ai)}$ in the form 
\begin{equation}
C^{(\Ai)}(z)=\widehat{C}^{(\Ai)}(z)e^{s\widehat{g}(z)\sigma_{3}}.\label{eq: formula for C^Ai}
\end{equation}
If we define an auxiliary matrix function (whose jumps are shown in
Fig. \ref{fig: Airy problem}) 
\begin{equation}
\begin{cases}
\Phi^{(\Ai)}= & \begin{pmatrix}-y_{1} & -y_{2}\\
-y_{1}' & -y_{2}'
\end{pmatrix}\;\mbox{in}\;II,\\
\Phi^{(\Ai)}= & \begin{pmatrix}-y_{2} & y_{1}\\
-y_{2}' & y_{1}'
\end{pmatrix}\;\mbox{in}\;III,
\end{cases}\qquad\begin{cases}
\Phi^{(\Ai)}= & \begin{pmatrix}y_{0} & -y_{2}\\
y_{0}' & -y_{2}'
\end{pmatrix}\;\mbox{in}\;I,\\
\Phi^{(\Ai)}= & \begin{pmatrix}y_{0} & y_{1}\\
y_{0}' & y_{1}'
\end{pmatrix}\;\mbox{in}\;IV,
\end{cases}\label{eq: solution Phi-Airy}
\end{equation}
where 
\begin{equation}
y_{0}(z)=\mathrm{Ai}(z),\quad y_{1}(z)=e^{2\pi i/3}\mathrm{Ai}(e^{2\pi i/3}z),\quad y_{2}(z)=e^{4\pi i/3}\mathrm{Ai}(e^{4\pi i/3}z),\label{eq: Airy model solution auxiliary defs}
\end{equation}
then a standard argument shows that $\widehat{C}^{(\Ai)}$ must have
the form 
\begin{equation}
\widehat{C}^{(\Ai)}(z)=\sqrt{2\pi}\begin{pmatrix}0 & -1\\
-i & 0
\end{pmatrix}\tau^{\sigma_{3}/4}\Phi^{(\Ai)}(\tau(z+1)),\label{eq:formula for C^Ai hat}
\end{equation}

This in particular implies a refined asymptotics for $C^{\left(\Ai\right)}$:

\begin{gather}
C^{(\Ai)}(z)=\left(I+\frac{m^{\Ai}}{z}+\mathcal{O}\left(\frac{1}{z^{2}}\right)\right)z^{\sigma_{3}/4}\frac{1}{\sqrt{2}}\begin{pmatrix}1 & -i\\
-i & 1
\end{pmatrix},\\
m^{\Ai}=\frac{\sigma_{3}}{4}+\frac{7i}{48s}\begin{pmatrix}0 & 1\\
0 & 0
\end{pmatrix}.\label{eq: def of m^Ai}
\end{gather}
Airy solutions like this are discussed in much detail, for example,
in \cite{DeKrMcVeZ99}.

\begin{figure}[h]
\centering{}\scalebox{0.8}{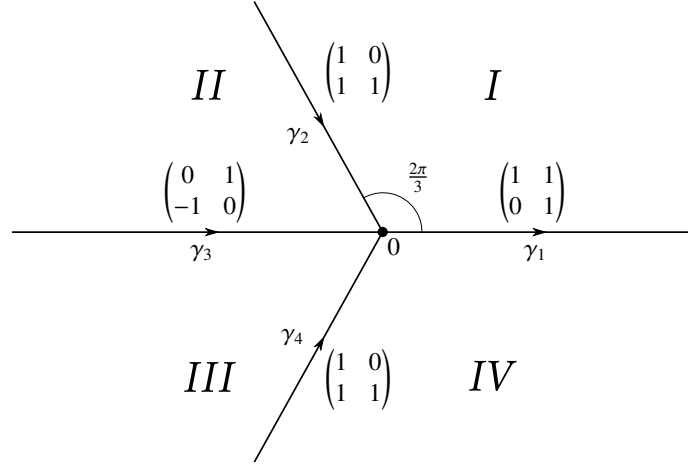}\caption{The standard Airy RH problem. The union of contours is referred to
as $\Gamma_{\Ai}$.\label{fig: Airy problem}}
\end{figure}

\subsubsection[Local solution \texorpdfstring{$C^{(0)}$}{C(0)}]{Local solution \boldmath{$C^{(0)}$}\label{sub: 4.3.2}}

We also need a local parametrix for $C$ near $z=0$. Let $U^{0}$
be a small open disk around $0$ of radius less than $1$, say, $1/2$.
Then we have to find the function $C^{(0)}$ which satisfies the following
RH problem. 
\begin{itemize}
\item[(a)] $C^{(0)}:\overline{U^{0}}\setminus[0,+\infty)\rightarrow\mathbb{C}^{2\times2}$
is analytic. 
\item[(b)] $C_{+}^{(0)}(z)=C_{-}^{(0)}(z)\begin{pmatrix}1 & (e^{-2\pi i\beta}-1)e^{-2s\widehat{g}(z)}\\
0 & 1
\end{pmatrix},\quad z\in(0,+\infty)\cap U^{0}$ (contour oriented to the right). 
\item[(c)] $C^{(0)}(z)=I+o\left(1\right)$ as $s\rightarrow\infty$, uniformly
on $\partial U^{0}$. 
\item[(d)] $C^{(0)}(z)\sim\left(I+\frac{\kappa^{2}}{2\pi i}\begin{pmatrix}0 & 1\\
0 & 0
\end{pmatrix}\log z\right)e^{s\widehat{g}(z)\sigma_{3}}$ as $z\rightarrow0$. The branch cut of the logarithm is along the
positive half axis.
\end{itemize}
Due to the simple algebraic structure of the jumps, this problem can
be solved exactly in terms of integrals of elementary functions. Namely,
the solution is 
\begin{equation}
C^{(0)}(z)=\begin{pmatrix}1 & {\displaystyle -\frac{\kappa^{2}}{2\pi i}\int_{0}^{1/2}\frac{e^{-2s\widehat{g}(z')}}{z'-z}\mathrm{d}z'}\\
0 & 1
\end{pmatrix},\;z\in\mathbb{C}\setminus\left[0,\frac{1}{2}\right].\label{eq: parametrix solution C^(0)+}
\end{equation}
This function clearly has the requested jumps and has the same general
logarithmic behavior near $z=0$. Moreover, this function satisfies
the matching condition on $\partial U^{0}$ and, in fact, with some
$c>0$ we have 
\begin{equation}
C^{(0)}(z)=I+\mathcal{O}\left(e^{-cs}\right)\mbox{ as }s\rightarrow\infty,\mbox{ uniformly on }\partial U^{0}.\label{eq: C^(0) is exponentially close to I}
\end{equation}

We will also need the fact that 
\begin{equation}
\lim_{z\rightarrow0}z\frac{\mathrm{d}C^{(0)}}{\mathrm{d}z}\left(C^{(0)}\right)^{-1}=\begin{pmatrix}0 & {\displaystyle \frac{\kappa^{2}}{2\pi i}e^{-2s\widehat{g}(0)}}\\
0 & 0
\end{pmatrix}.\label{eq: limit at z-0 for C^(0)}
\end{equation}

\subsection{Final transformation}

Using the functions built in the previous subsection, we can now perform
the final transformation of the RH analysis in the case where $\tau\rightarrow+\infty$.

Define 
\begin{equation}
D(z)=\begin{cases}
C(z)\cdot\left(C^{(0)}(z)\right)^{-1}\cdot\left(C^{(\Ai)}(z)\right)^{-1}, & z\in U^{0}\setminus\mathbb{R},\\
C(z)\cdot\left(C^{(\Ai)}(z)\right)^{-1}, & z\in\mathbb{C}\setminus\overline{U^{0}\cup\Gamma_{B}}.
\end{cases}\label{eq: def of D+}
\end{equation}
This function has the following properties. 
\begin{itemize}
\item[(a)] $D:\mathbb{C}\setminus\left(\left[\frac{1}{2},+\infty\right)\cup\partial U^{0}\right)\rightarrow\mathbb{C}^{2\times2}$
is analytic. 
\item[(b)] Assuming the counterclockwise orientation of $\partial U^{0}$, 
\begin{equation}
D_{+}(z)=D_{-}(z)\cdot\begin{cases}
C_{-}^{(\Ai)}(z)\cdot\begin{pmatrix}1 & -\kappa^{2}e^{-2s\widehat{g}(z)}\\
0 & 1
\end{pmatrix}\cdot\left(C_{-}^{(\Ai)}(z)\right)^{-1}, & z\in\left(\frac{1}{2},+\infty\right),\\
C^{(\Ai)}(z)\cdot\left(C^{(0)}(z)\right)^{-1}\cdot\left(C^{(\Ai)}(z)\right)^{-1}, & z\in\partial U^{0}.
\end{cases}\label{eq: jumps of D+}
\end{equation}

\item[(c)] $D(z)=I+\mathcal{O}\left(\frac{1}{z}\right)$, as $z\rightarrow\infty$. 
\end{itemize}
Using the asymptotic expansion for $C^{(\Ai)}$, it is easy to check
that, with some $c>0$, 
\begin{equation}
\left(D_{-}(z)\right)^{-1}D_{+}(z)=I+\mathcal{O}\left(e^{-cs|z|}\right)\mbox{ as }s\rightarrow\infty,\mbox{ uniformly for }z\in\left(\frac{1}{2},+\infty\right).\label{eq:jumps of D are exponentially close to I}
\end{equation}
As for the jump on $\partial U^{0}$, by virtue of (\ref{eq: C^(0) is exponentially close to I})
and the boundedness of $C^{(\Ai)}$, it is also close to the identity
matrix: 
\begin{gather}
\left(D_{-}(z)\right)^{-1}D_{+}(z)=C^{(\Ai)}(z)\cdot\left(I+\mathcal{O}\left(e^{-cs}\right)\right)\cdot\left(C^{(\Ai)}(z)\right)^{-1}=I+\mathcal{O}\left(e^{-cs}\right)\label{eq: jumps of D+ are close to 1}\\
\mbox{ as }s\rightarrow\infty,\mbox{ uniformly for }z\in\partial U^{0},\mbox{ with some }c>0.\nonumber 
\end{gather}

Using these estimates, in a standard way \cite{DeiftZhou93} one shows
that, for any $z\in\mathbb{C}\setminus\Gamma_{D}$, 
\begin{equation}
D(z)=I+\mathcal{O}\left(\frac{e^{-cs}}{1+|z|}\right)\mbox{ as }s\rightarrow\infty,\ c>0.\label{eq: estimate for D+}
\end{equation}
The error term is uniform on compact subsets of $\mathbb{C}\setminus\Gamma_{D}$.

\subsection[Asymptotics for $y$ and uniformity of error terms]{Asymptotics for \boldmath{$y$} and uniformity of error terms}

Following the transformations $\Phi\mapsto\Psi_{0}\mapsto A\mapsto B\mapsto C\mapsto D$
(eqs. \eqref{def Phi}, \eqref{eq: def of A+}, \eqref{eq: def of B+},
\eqref{eq: def of C+}, \eqref{eq: def of D+}) backwards, we can
recover the connection between the asymptotic expansions of $\Phi$
and $D$. Namely, for large $z\in\mathbb{C}\setminus\left(II'\cup III'\right)$,
we have 
\begin{equation}
\Phi(z)=\tau^{\sigma_{3}/4}D\left(\frac{z}{\tau}-1\right)\sqrt{2\pi}\begin{pmatrix}0 & -1\\
-i & 0
\end{pmatrix}\tau^{\sigma_{3}/4}\Phi^{(\Ai)}(z).\label{eq: Phi in terms of D+}
\end{equation}
Next, we write, as usual, 
\begin{gather}
D(z)=I+\frac{m^{D}}{z}+\frac{m^{D;2}}{z^{2}}+\mathcal{O}\left(\frac{1}{z^{3}}\right)\mbox{ as }z\rightarrow\infty,\\
\Phi(z)=\left(I+\frac{m^{\Phi}}{z}+\frac{m^{\Phi;2}}{z^{2}}+\mathcal{O}\left(\frac{1}{z^{3}}\right)\right)z^{\sigma_{3}/4}\frac{1}{\sqrt{2}}\begin{pmatrix}1 & -i\\
-i & 1
\end{pmatrix}e^{-\frac{2}{3}z^{3/2}\sigma_{3}}\mbox{ as }z\rightarrow\infty,\label{eq: new Phi expansion}
\end{gather}
Let us refer to the matrix coefficients in front of $z^{-k}$ in these
expansions as $m^{D;k}$ and $m^{\Phi;k}$ (they are only functions
of $\tau$). Using \eqref{eq: Phi in terms of D+}, we see that each
of the matrices $m^{\Phi;k}$ in the expansion for $\Phi$ is merely
a linear combination of a finite number of the matrix coefficients
$m^{D;k}$, with coefficients rational in $\tau^{1/4}$.

Using \eqref{eq: estimate for D+}, one shows that the matrices $m^{D;k}$
are exponentially small, 
\begin{equation}
m^{D;k}(\tau)=\mathcal{O}\left(e^{-c\tau^{3/2}}\right)\mbox{ as }\tau\rightarrow+\infty,\mbox{ with some }c>0
\end{equation}
for all $k$. It immediately follows that 
\begin{equation}
m^{\Phi;k}(\tau)=\widetilde{m}^{\Ai;k}+\mathcal{O}\left(\tau^{k+\frac{1}{2}}e^{-c\tau^{3/2}}\right)\mbox{ as }\tau\rightarrow+\infty,
\end{equation}
thus $m^{\Phi;k}$ are bounded at large $\tau$.


These facts imply that the asymptotic expansion (\ref{eq: asymptotic condition for Phi(z)})
for $\Phi$ is \textit{uniform} for $\tau\in\left[\tau_{0},+\infty\right)$
for any $\tau_{0}\in\mathbb{R}$.

Since $C(z)=D(z)\cdot C^{(\Ai)}(z)\cdot C^{(0)}(z)$ in $U^{0}$ and
both $D(z)$ and $C^{(\Ai)}$ are bounded there, we get from (\ref{eq: y(tau) in terms of C(z)})
that 
\begin{equation}
y(\tau)=\frac{i}{\sqrt{\tau}}\lim_{z\rightarrow0}\left[z\frac{\mathrm{d}C(z)}{\mathrm{d}z}C(z)^{-1}\right]_{21}=\frac{\kappa^{2}e^{-2s\widehat{g}(0)}}{2\pi\sqrt{\tau}}\left[D(0)C^{(\Ai)}(0)\begin{pmatrix}0 & 1\\
0 & 0
\end{pmatrix}\left(D(0)C^{(\Ai)}(0)\right)^{-1}\right]_{21}\mbox{ for large }\tau>0.\label{eq: y(tau) in terms of D+}
\end{equation}
From (\ref{eq: formula for C^Ai}), (\ref{eq:formula for C^Ai hat}),
and the asymptotics for $\Phi^{(\Ai)}$ it follows that 
\begin{equation}
C^{(\Ai)}(0)=\sqrt{2\pi}\begin{pmatrix}0 & -1\\
-i & 0
\end{pmatrix}\tau^{\sigma_{3}/4}\Phi^{(\Ai)}(\tau)e^{\frac{2}{3}\tau^{3/2}\sigma_{3}}=\frac{1}{\sqrt{2}}\begin{pmatrix}1 & -i\\
-i & 1
\end{pmatrix}+\mathcal{O}\left(\frac{1}{\tau^{3/2}}\right)\mbox{ as }\tau\rightarrow+\infty.\label{eq: C^(Ai)(0) asymptotic}
\end{equation}
Taking into account that $D(0)$ converges to $I$ very rapidly and
using the definitions (\ref{eq:g_m-function}) and (\ref{eq:definition of s+}),
we arrive at the final expression for the asymptotic behavior of $y(\tau;\beta)$,
where we emphasize the dependence on $\beta$:

\begin{equation}
y(\tau;\beta)=e^{-\frac{4}{3}\tau^{3/2}}\left(\frac{\kappa^{2}}{4\pi\sqrt{\tau}}+\mathcal{O}\left(\frac{1}{\tau^{2}}\right)\right)=\left(\kappa\Ai\left(\tau\right)\right)^{2}\left(1+\mathcal{O}\left(\tau^{-2/3}\right)\right)\mbox{ as }\tau\rightarrow+\infty.\label{eq: asymptotic for y(tau) at +Infty}
\end{equation}

Together with $y(\tau;\beta)=u(\tau;\kappa)^{2}$, this proves \eqref{u+}.

\section{Asymptotics of the recurrence coefficients}

\label{section: asymptotics recur}

In this section, we compute asymptotics for the recurrence coefficients
$R_{n}$ and $Q_{n}$. Our calculations in this section are similar
to those in \cite{XuZhao11}, but we believe it is convenient for
the reader to give some details of the calculations because of differences
in notations.

\subsection[Auxiliary asymptotics of $G_{R}$]{Auxiliary asymptotics of \boldmath{$G_{R}$}}

We now need to compute the precise asymptotic behavior of $G_{R}$,
the jump matrix for $R$ (see \eqref{eq: G_R is everywhere close to I}).
Finding an explicit expression for the two leading terms in $G_{R}$
on $\partial U^{1}$ is the most sophisticated part of this calculation.
First, expand 
\begin{equation}
G_{R}(z)=P^{(\infty)}(z)\left(P^{(1)}(z)\right)^{-1}=P^{(\infty)}(z)e^{-\frac{2}{3}\zeta(z)^{3/2}\sigma_{3}}e^{i\pi\beta\sigma_{3}/2}\Phi\left(\zeta(z)\right)^{-1}E(z)^{-1}\mbox{ for }z\in\partial U^{1}\label{eq: G_R in terms of Psi}
\end{equation}
with $E(z)$ defined in (\ref{eq: definition of E(z)}). Recall that
we have the asymptotic expansion \eqref{eq: asymptotic condition for Phi(z)}
for $\Phi$, uniformly for $\tau\geqslant\tau_{0}$ with any fixed
$\tau_{0}\in\mathbb{R}$. Therefore, one verifies using \eqref{eq: definition of P(infty)}
that, as $n\to\infty$, 
\begin{multline}
G_{R}(z)=e^{i\pi\beta\sigma_{3}/2}\frac{1}{\sqrt{2}}\begin{pmatrix}1 & i\\
i & 1
\end{pmatrix}\left(\frac{z-1}{z+1}\right)^{\sigma_{3}/4}\left(I-\frac{m_{21}^{\Phi}}{\zeta(z)^{1/2}}\begin{pmatrix}0 & 0\\
1 & 0
\end{pmatrix}-\frac{m_{11}^{\Phi}}{\zeta(z)}\begin{pmatrix}1 & 0\\
0 & -1
\end{pmatrix}+\mathcal{O}\left(n^{-1}\right)\right)\;\times\\
\times\;\left(\frac{z-1}{z+1}\right)^{-\sigma_{3}/4}\frac{1}{\sqrt{2}}\begin{pmatrix}1 & -i\\
-i & 1
\end{pmatrix}e^{-i\pi\beta\sigma_{3}/2},
\end{multline}
which gives us the following expansion of $G_{R}$ as $n\to\infty$:
\begin{equation}
G_{R}(z)=I-G_{1}(z)n^{-1/3}+G_{2}(z)n^{-2/3}+\mathcal{O}\left(n^{-1}\right),\mbox{ uniformly for }z\in\partial U^{1}\mbox{ and }\tau\geqslant\tau_{0},\label{eq:expansion of G_R part 3}
\end{equation}
where 
\begin{equation}
G_{1}(z)=\frac{im_{21}^{\Phi}\left(z+1\right)^{1/2}}{2}\frac{n^{1/3}}{\left(z-1\right)^{1/2}\zeta(z)^{1/2}}\begin{pmatrix}1 & -ie^{i\pi\beta}\\
-ie^{-i\pi\beta} & -1
\end{pmatrix}\label{eq:def of G_1(z)}
\end{equation}
and 
\begin{equation}
G_{2}(z)=\frac{im_{11}^{\Phi}n^{2/3}}{\zeta(z)}\begin{pmatrix}0 & e^{i\pi\beta}\\
-e^{i\pi\beta} & 0
\end{pmatrix}.\label{eq:def of G_2(z)}
\end{equation}

\subsection[Asymptotics of $R_{n}$]{Asymptotics of \boldmath{$R_{n}$}}

We can use the following simple identity for the recurrence coefficient
$R_{n}$ defined in (\ref{recur}) (see e.g. \cite{Deift99}): 
\begin{equation}
R_{n}=m_{12}^{Y}m_{21}^{Y},\label{eq: R_n}
\end{equation}
where the matrix $m^{Y}$ is defined in terms of the large $z$ expansion
of $Y$: 
\begin{equation}
Y(z)=\left(I+\frac{m^{Y}(t)}{z}+\mathcal{O}\left(\frac{1}{z^{2}}\right)\right)z^{n\sigma_{3}}.\label{eq:definition of mY}
\end{equation}
In order to compute $m^{Y}$, we will need similar large $z$ expansions
for the following functions 
\begin{align}
R(z) & =I+\frac{m^{R}(t)}{z}+\mathcal{O}\left(\frac{1}{z^{2}}\right),\label{eq:def of mR}\\
P^{(\infty)} & =I+\frac{m^{\infty}(t)}{z}+\mathcal{O}\left(\frac{1}{z^{2}}\right),\label{eq:def of m-Infty}\\
g(z) & =\log z-\frac{1}{8z^{2}}+\mathcal{O}\left(\frac{1}{z^{4}}\right).\label{eq:expansion of g(z)}
\end{align}
Unfolding the transformations $Y\mapsto T\mapsto S\mapsto R$ at large
$z$, we obtain the identity 
\begin{equation}
m^{Y}=\sqrt{2n}e^{nl\sigma_{3}/2}\left(2n\right)^{n\sigma_{3}/2}\left(m^{R}+m^{\infty}\right)\left(2n\right)^{-n\sigma_{3}/2}e^{-nl\sigma_{3}/2}.\label{eq: mY in terms of mInf and mR}
\end{equation}

Since we can reformulate the RH problem for $R$ in terms of an integral
equation 
\begin{equation}
R_{-}(z)=I+\frac{1}{2\pi i}\int_{\Gamma_{R}}\frac{R_{-}(z)\left(G_{R}(z')-I\right)}{z'-z}\mathrm{d}z',\label{eq: integral equation for R}
\end{equation}
we have 
\begin{equation}
m^{R}=-\frac{1}{2\pi i}\int_{\Gamma_{R}}R_{-}(z')\left(G_{R}(z')-I\right)\mathrm{d}z'.\label{eq: mR part 1}
\end{equation}
Next we iterate the integral equation to find an asymptotic expansion
for $R_{-}$ as $n\to\infty$. Given that integration over the contours
other than $\partial U^{1}$ gives only a $\mathcal{O}\left(n^{-1}\right)$
contribution (because the jump matrix $G_{R}=I+\bigO(n^{-1})$ on
$\Gamma_{R}\setminus\partial U^{1}$), the large $n$ expansion for
$m^{R}$ is 
\begin{equation}
m^{R}=-\frac{1}{2\pi i}\ointctrclockwiseop_{\partial U^{1}}\left(G_{R}(z')-I\right)\mathrm{d}z'-\frac{1}{2\pi i}\ointctrclockwiseop_{\partial U^{1}}\rho_{1}(z')\left(G_{R}(z')-I\right)\mathrm{d}z'+\mathcal{O}\left(n^{-1}\right).
\end{equation}

We can now substitute the asymptotics (\ref{eq:expansion of G_R part 3})
to get, after a straightforward calculation, 
\begin{equation}
m^{R}=n^{-1/3}\res_{z=1}G_{1}(z)-n^{-2/3}\res_{z=1}G_{2}(z)+n^{-2/3}\res_{z=1}G_{1}(z)\cdot\res_{z=1}\frac{G_{1}(z)}{z-1}+\mathcal{O}\left(n^{-1}\right),\label{eq: mR part 3}
\end{equation}
as $n\to\infty$. Now, from (\ref{eq:def of G_1(z)}) and the expansions
\eqref{eq: expansion zeta}, \eqref{eq: expansion tau} we find that,
as $n\to\infty$, 
\begin{equation}
\res_{z=1}G_{1}(z)=\frac{im_{21}^{\Phi}(t)}{2}\begin{pmatrix}1 & -ie^{i\pi\beta}\\
-ie^{-i\pi\beta} & -1
\end{pmatrix}+\mathcal{O}\left(n^{-2/3}\right)
\end{equation}
and 
\begin{equation}
\res_{z=1}\frac{G_{1}(z)}{z-1}=\frac{im_{21}^{\Phi}(t)}{10}\begin{pmatrix}1 & -ie^{i\pi\beta}\\
-ie^{-i\pi\beta} & -1
\end{pmatrix}+\mathcal{O}\left(n^{-2/3}\right),
\end{equation}
as well as, from (\ref{eq:def of G_2(z)}), 
\begin{equation}
\res_{z=1}G_{2}(z)=\frac{im_{11}^{\Phi}(t)}{2}\begin{pmatrix}0 & e^{i\pi\beta}\\
-e^{-i\pi\beta} & 0
\end{pmatrix}+\mathcal{O}\left(n^{-2/3}\right).
\end{equation}
Note that $\res_{z=1}G_{1}$ is nilpotent, thus the third term in
(\ref{eq: mR part 3}) is negligible. Furthermore, from the relations
(\ref{eq: m in terms of mPhi}) and (\ref{eq:m11 through m21}), we
find 
\begin{equation}
m_{11}^{\Phi}=\frac{1}{2}\left(m_{21}^{\Phi}\right)^{2}-\frac{i}{2}\left(m_{21}^{\Phi}\right)'.\label{eq: m_11^Phi in terms of m_21^Phi}
\end{equation}
Substituting all previous results into (\ref{eq: mR part 3}) we obtain
the final formula 
\begin{equation}
m^{R}=\frac{im_{21}^{\Phi}(t)}{2}\begin{pmatrix}1 & -ie^{i\pi\beta}\\
-ie^{-i\pi\beta} & -1
\end{pmatrix}n^{-1/3}-\frac{im_{11}^{\Phi}(t)}{2}\begin{pmatrix}0 & e^{i\pi\beta}\\
-e^{-i\pi\beta} & 0
\end{pmatrix}n^{-2/3}+\mathcal{O}\left(n^{-1}\right).\label{eq: mR final}
\end{equation}
The second matrix in \eqref{eq: mY in terms of mInf and mR}, $m^{\infty}$,
can be easily found from (\ref{eq: definition of P(infty)}): 
\begin{equation}
m^{\infty}=\frac{i}{2}\begin{pmatrix}0 & e^{i\pi\beta}\\
-e^{-i\pi\beta} & 0
\end{pmatrix}.\label{eq: m-infty}
\end{equation}
All the operations performed to obtain the asymptotics of $m^{R}$
from the asymptotics of $\Phi$ preserve the uniformity in $\tau\in\left[\tau_{0},\infty\right)$,
or equivalently $t\in\left[t_{0},\infty\right)$, for any $\tau_{0},t_{0}\in\mathbb{R}$.

By substituting $m^{R}$ and $m^{\infty}$ into (\ref{eq: mY in terms of mInf and mR})
we find the large $n$ expansion for $R_{n}$, 
\begin{multline}
R_{n}=m_{12}^{Y}m_{21}^{Y}=2n\left(m_{12}^{R}+m_{12}^{\infty}\right)\left(m_{21}^{R}+m_{21}^{\infty}\right)=\\
=2n\left(\frac{m_{21}^{\Phi}}{2}n^{-1/3}-\frac{im_{11}^{\Phi}}{2}n^{-2/3}+\frac{i}{2}+\mathcal{O}\left(n^{-1}\right)\right)\left(\frac{m_{21}^{\Phi}}{2}n^{-1/3}+\frac{im_{11}^{\Phi}}{2}n^{-2/3}-\frac{i}{2}+\mathcal{O}\left(n^{-1}\right)\right),
\end{multline}
which, by (\ref{eq: m_11^Phi in terms of m_21^Phi}), simplifies to
\begin{equation}
R_{n}(\lambda_{0};\beta)=\frac{n}{2}-\frac{y(t;\beta)}{2}n^{1/3}+\mathcal{O}\left(1\right)\mbox{ as }n\rightarrow\infty,\mbox{ for all }t\in\mathbb{R},\ \mbox{ uniformly for }t\geqslant t_{0},\label{eq: final R_n asymptotics}
\end{equation}
since $\left(-im_{21}^{\Phi}\right)'=y$. This result holds for all
$t\in\mathbb{R}$ and $\beta$ such that $|\Re\beta|<1/2$. This asymptotic
series formally matches the classical Hermite recurrence coefficient
asymptotics when $\lambda>1$ ($t\rightarrow+\infty)$ and the non-critical
asymptotics from \cite{ItsKras08} when $\lambda<1$ ($t\rightarrow-\infty)$.

\subsection[Asymptotics of \texorpdfstring{$p_{n}\left(\lambda_{0}\right)$}{pn(lambda0)}
and $Q_{n}$]{Plancherel-Rotach type formula and asymptotics for \boldmath{$Q_{n}$}}

We can express the orthogonal polynomial $p_{n}$ in terms of the
RH solution $Y$, 
\begin{equation}
p_{n}\left(\lambda\sqrt{2n}\right)=Y_{11}\left(\lambda\sqrt{2n}\right)=\lim_{z\rightarrow\lambda}\left(2n\right)^{n/2}S_{11}\left(z\right)e^{ng_{+}\left(z\right)},
\end{equation}
where the limit for $S$ is taken for $z$ approaching $\lambda=1+\frac{t}{2}n^{-2/3}$
from the upper half plane, outside the lens-shaped region in Fig.
\ref{fig:The-opening-of lenses}. If $z$ lies in this region and
$z\in U^{1}$ (the small disk around $1$ in which the local parametrix
$P^{(1)}$ was constructed), then we can unwind the transformations
$S\mapsto P^{(1)}\mapsto\Phi\mapsto\Psi_{0}$ to obtain 
\begin{equation}
S(z)=R\left(z\right)P^{(1)}\left(z\right)=R\left(z\right)E\left(z\right)\begin{pmatrix}1 & {\displaystyle -\frac{i\tau^{2}}{4}}\\
0 & 1
\end{pmatrix}\Psi_{0}\left(\zeta\left(z\right)-\tau\right)e^{-i\pi\beta\sigma_{3}/2}e^{\frac{2}{3}\zeta(z)^{3/2}\sigma_{3}}.
\end{equation}
In order to compute the limit $z\to\lambda$, we need to use the small
$\xi$ expansion of $\Psi_{0}$ in sector $I$ given in \eqref{eq: model behavior at zero}.
After a straightforward calculation, using also \eqref{asymptotics R}
and \eqref{eq: expansion tau}, we get 
\begin{equation}
\lim_{z\rightarrow\lambda}S_{11}\left(z\right)e^{ng_{+}\left(z\right)}=in^{1/6}c\left(t\right)e^{-\frac{2}{3}is+ng_{+}(\lambda)}\left(1+\mathcal{O}\left(n^{-1/3}\right)\right),\qquad n\to\infty.
\end{equation}
Therefore, 
\begin{equation}
p_{n}\left(\lambda\sqrt{2n}\right)=\left(2n\right)^{n/2}e^{-\frac{2}{3}is+ng_{+}(\lambda)}ic(t)n^{1/6}\left(1+\mathcal{O}\left(n^{-1/3}\right)\right),\qquad n\to\infty.
\end{equation}
From the Lax pair identity (\ref{eq: connection c^2 to m_21}) and
$y(t;\beta)=u(t;\kappa)^{2}$, 
\begin{equation}
\left(ic\left(t;\beta\right)\right)^{2}=\frac{2\pi u\left(\tau;\kappa\right)^{2}}{\kappa^{2}}\label{relation c u 1}
\end{equation}
and 
\begin{equation}
ic\left(t;\beta\right)=\pm\frac{\sqrt{2\pi}u\left(\tau;\kappa\right)}{\kappa}.\label{sign}
\end{equation}
The right hand side does not depend on the sign of $\kappa$ (indeed,
changing $\kappa$ to $-\kappa$ changes $u$ to $-u$), and we can
verify which sign is correct using the asymptotics for $c$ as $t\to\infty$.
Since 
\[
c\left(\tau\right)=\lim_{z\rightarrow0}\left(\Psi_{0}\left(z\right)\right)_{21},
\]
working backwards along the transformations $\Psi_{0}\mapsto A\mapsto B\mapsto C\mapsto D$
for both $\tau\rightarrow\pm\infty$, we can easily recover the asymptotics
for $c\left(\tau\right)$. It turns out that 
\begin{equation}
ic\left(\tau;\beta\right)=\frac{e^{-\frac{2}{3}\tau^{3/2}}}{\sqrt{2}\tau^{1/4}}\left(1+\mathcal{O}\left(\tau^{-2}\right)\right)\mbox{ as }\tau\rightarrow+\infty,\label{c as}
\end{equation}
which implies that the correct sign in \eqref{sign} is $+$. 
\begin{rem}
In the special case $\beta=0$, the model RH problem for $\Psi_{0}$
reduces to the Airy model RH problem. In this case, we have 
\begin{equation}
ic(\tau;\beta=0)=\sqrt{2\pi}\Ai(\tau),\label{c Airy}
\end{equation}
which is indeed consistent with \eqref{c as}. 
\end{rem}
We thus have 
\begin{equation}
p_{n}\left(\lambda\sqrt{2n}\right)=\frac{\sqrt{2\pi}}{\kappa}(2n)^{n/2}e^{-\frac{2}{3}is+ng_{+}(\lambda)}u(t;\kappa)n^{1/6}\left(1+\mathcal{O}\left(n^{-1/3}\right)\right),\qquad n\to\infty.
\end{equation}

Now we need the expansion of $g(z)$ near $z=1$: 
\begin{equation}
g(z)=\frac{1}{2}-\log2+2\left(z-1\right)-\frac{2}{3}2^{3/2}\left(z-1\right)^{3/2}+\mathcal{O}\left(z-1\right)^{2}\mbox{ as }z\rightarrow1.
\end{equation}
Substituting $\lambda=1+\frac{t}{2}n^{-2/3}$, we have 
\begin{equation}
2ng_{+}\left(\lambda\right)=n-2n\log2+2tn^{1/3}+\frac{4}{3}is+\mathcal{O}\left(n^{-1/3}\right)\mbox{ as }n\rightarrow\infty.
\end{equation}

This gives us the asymptotics as $n\to\infty$ of the polynomials
$p_{n}$ near the critical point, 
\begin{equation}
p_{n}\left(\lambda\sqrt{2n}\right)=\frac{\sqrt{2\pi}}{\kappa}\left(\frac{ne}{2}\right)^{n/2}n^{1/6}e^{tn^{1/3}}u\left(t;\kappa\right)\left(1+\mathcal{O}\left(n^{-1/3}\right)\right).\label{eq: Plancherel-Rotach type formula}
\end{equation}

By multiplying the recurrence relation (\ref{recur}) by $p_{n}(x)w(x)$
and integrating, we find 
\begin{equation}
Q_{n}=-h_{n}^{-1}p_{n}\left(\lambda\sqrt{2n}\right)^{2}e^{-2n\lambda^{2}}\sinh\left(i\pi\beta\right).\label{id Qn}
\end{equation}
Note that 
\begin{multline}
h_{n}=-\lim_{z\rightarrow\infty}2\pi i\,Y_{21}\left(z\sqrt{2n}\right)\left(z\sqrt{2n}\right)^{n+1}=-\lim_{z\rightarrow\infty}2\pi i\sqrt{2n}\left(2n\right)^{n}e^{nl}z\,S_{12}(z)=\\
=-2\pi i\sqrt{2n}\left(2n\right)^{n}e^{nl}\left(m_{12}^{\infty}+m_{12}^{R}\right),
\end{multline}
thus the following large $n$ asymptotics hold for the normalizing
coefficients $h_{n}$ 
\begin{gather}
h_{n}=\frac{\pi\sqrt{2n}n^{n}}{2^{n}e^{n}}e^{i\pi\beta}\left(1-im_{21}^{\Phi}(t)n^{-1/3}-m_{11}^{\Phi}(t)n^{-2/3}+\mathcal{O}\left(n^{-1}\right)\right).\label{as hn}
\end{gather}
This proves \eqref{eq: h-n asymptotics}. Equivalently, this result
can be deduced from the identity 
\begin{equation}
h_{n}=\frac{H_{n+1}}{H_{n}},
\end{equation}
expressing $h_{n}$ as a ratio of two Hankel determinants, together
with the asymptotics \eqref{hankelas}. Substituting \eqref{as hn}
and \eqref{PR as} in \eqref{id Qn}, we obtain (\ref{Qnas}).

Lastly, we note that we can easily obtain the asymptotics of the coefficients
in \eqref{as hn} as $t\to-\infty$. Formally this can be done by
computing an antiderivative of the asymptotics of $y\left(t\right)$.
The following asymptotics were obtained rigorously by solving the
RH problem in the limit $t\to-\infty$. For $\left|\Re\beta\right|<1/2$,

\begin{multline}
-im_{21}^{\Phi}\left(\tau;\beta\right)=-2i\beta\sqrt{-\tau}-\frac{1}{4i\left(-\tau\right)}\left(\frac{\Gamma(1-\beta)}{\Gamma(\beta)}e^{i\theta(\tau;\beta)}-\frac{\Gamma(1+\beta)}{\Gamma(-\beta)}e^{-i\theta(\tau;\beta)}\right)-\\
-\frac{3\beta^{2}}{2\left(-\tau\right)}+\mathcal{O}\left(\frac{1}{\left(-\tau\right)^{5/2-3\left|\Re\beta\right|}}\right),\mbox{ as }\tau\rightarrow-\infty.\label{eq: m21Phi asymptotic for general beta-1}
\end{multline}

When $\beta=i\kappa$, $\kappa\in\mathbb{R}$, this becomes

\begin{multline}
-im_{21}^{\Phi}\left(\tau;i\kappa\right)=2\kappa\sqrt{-\tau}+\frac{\kappa}{2\left(-\tau\right)}\cos\left(\frac{4}{3}(-\tau)^{3/2}+3\kappa\log(-\tau)+6\kappa\log2-2\arg\Gamma(i\kappa)\right)+\\
+\frac{3\kappa^{2}}{2\left(-\tau\right)}+\mathcal{O}\left(\frac{1}{\left(-\tau\right)^{5/2}}\right),\mbox{ as }\tau\rightarrow-\infty.\label{eq: asymptotic of -im21^Phi real case-1}
\end{multline}

When $\beta=1/2+i\gamma$, $\gamma\in\mathbb{R}$, we have
\begin{equation}
-im_{21}^{\Phi}\left(\tau;\frac{1}{2}+i\gamma\right)=\sqrt{-\tau}\left(2\gamma-\tg\left(\frac{\widetilde{\theta}}{2}\right)\right)+\mathcal{O}\left(\frac{1}{\tau}\right),\mbox{ as }\tau\rightarrow-\infty.\label{eq: m21Phi asymptotic for 1/2-1}
\end{equation}
These formulas complement Theorem \ref{theorem rec}.

\section{Hankel determinants: alternative proof of Theorem \ref{theorem hankel}}

\label{section: Hankel}

\subsection{Differential identity}

Here, we derive a differential identity for the logarithm of the Hankel
determinant $H_{n}(\lambda_{0},\beta)$. It is expressed in terms
of $Y$ defined in \eqref{eq: explicit solution of the OP RH problem}. 
\begin{prop}
We have 
\begin{equation}
\frac{\partial}{\partial\lambda_{0}}\log H_{n}(\lambda_{0},\beta)=\frac{1}{\pi}\sin\pi\beta\,\left(Y^{-1}Y'\right)_{21}(\lambda_{0})e^{-\lambda_{0}^{2}}.\label{diff id}
\end{equation}
Here $'$ is the derivative of $Y(z)$ with respect to $z$. \end{prop}
\begin{proof}
We write $P_{k}=\kappa_{k}p_{k}$, $\kappa_{k}=\frac{1}{\sqrt{h_{k}}}>0$
for the normalized orthogonal polynomials with respect to the weight
$w$. We start from the general identity (equation (17) in \cite{Krasovsky07})
\begin{equation}
\frac{\partial}{\partial\lambda_{0}}\log H_{n}(\lambda_{0},n)=-n\frac{\dot{\kappa}_{n-1}}{\kappa_{n-1}}+\frac{\kappa_{n-1}}{\kappa_{n}}\left(J_{1}-J_{2}\right),\label{diff id1}
\end{equation}
where 
\begin{align}
J_{1} & =\int_{\mathbb{R}}\dot{P}_{n}(x)P_{n-1}'(x)w(x)\mathrm{d}x,\label{def J1}\\
J_{2} & =\int_{\mathbb{R}}P_{n}'(x)\dot{P}_{n-1}(x)w(x)\mathrm{d}x.\label{def J2}
\end{align}
Here and below dots denote $\lambda_{0}$-derivatives and primes denote
$x$-derivatives.

To compute $J_{1}$, we proceed as follows: by (\ref{def J1}) and
(\ref{weight}), we have 
\begin{multline}
J_{1}=\frac{\partial}{\partial\lambda_{0}}\left(\int_{\mathbb{R}}P_{n}(x)P_{n-1}'(x)w(x)\mathrm{d}x\right)-\int_{\mathbb{R}}P_{n}(x)\dot{P}_{n-1}'(x)w(x)\mathrm{d}x\\
+2i\sin(\pi\beta)P_{n}(\lambda_{0})P_{n-1}'(\lambda_{0})e^{-\lambda_{0}^{2}}.
\end{multline}
The first two terms vanish by orthogonality, and we obtain 
\begin{equation}
J_{1}=2i\sin(\pi\beta)(P_{n}P_{n-1}')(\lambda_{0})e^{-\lambda_{0}^{2}}.\label{J1}
\end{equation}

Similarly, by (\ref{def J2}) and (\ref{weight}), we have 
\begin{equation}
J_{2}=\frac{\partial}{\partial\lambda_{0}}\left(\int_{\mathbb{R}}P_{n}'(x)P_{n-1}(x)w(x)dx\right)-\int_{\mathbb{R}}\dot{P}_{n}'(x)P_{n-1}(x)w(x)dx+2i\sin(\pi\beta)P_{n}'(\lambda_{0})P_{n-1}(\lambda_{0})e^{-\lambda_{0}^{2}}.
\end{equation}
Using the orthogonality relations, we can compute the first two terms
and we get 
\begin{equation}
J_{2}=-n\frac{\kappa_{n}}{\kappa_{n-1}^{2}}\dot{\kappa_{n-1}}+2i\sin(\pi\beta)\left(P_{n}'P_{n-1}\right)(\lambda_{0})e^{-\lambda_{0}^{2}}.\label{J2}
\end{equation}

Substituting (\ref{J1}) and (\ref{J2}) into (\ref{diff id1}), we
get 
\begin{eqnarray}
\frac{\partial}{\partial\lambda_{0}}\log H_{n}(\lambda_{0},\beta) & = & \frac{2i\kappa_{n-1}}{\kappa_{n}}\left(P_{n}P_{n-1}'-P_{n}'P_{n-1}\right)(\lambda_{0})\sin(\pi\beta)e^{-\lambda_{0}^{2}}\\
 & = & \frac{2i}{h_{n-1}}\left(p_{n}p_{n-1}'-p_{n}'p_{n-1}\right)(\lambda_{0})\sin(\pi\beta)e^{-\lambda_{0}^{2}},\label{diff id0}
\end{eqnarray}
and using (\ref{eq: explicit solution of the OP RH problem}), we
obtain (\ref{diff id}). 
\end{proof}

\subsection[Asymptotics for the log. derivative of \texorpdfstring{$H_{n}$}{H\_n(lambda0,b)}]{Asymptotics for the logarithmic derivative of \boldmath{$H_{n}\left(\lambda_{0},\beta\right)$}}

Let $\lambda_{0}$ be of the form \eqref{scaling lambda0}. The results
in Section \ref{section: RH OP} are valid in the limit where $n\to\infty$,
uniformly for $t\geq t_{0}$ for any fixed $t_{0}\in\mathbb{R}$.

Inverting the transformations $Y\mapsto T$ and $T\mapsto S$ from
Section \ref{sub:Overview-of-transformations}, it follows from (\ref{diff id})
that 
\begin{equation}
\frac{\partial}{\partial\lambda_{0}}\ln H_{n}(\lambda_{0},\beta)=\frac{\sin\left(\pi\beta\right)}{\pi\sqrt{2n}}\,\lim_{z\to\lambda}\left(S^{-1}(z)S'(z)\right)_{21}
\end{equation}
and the limit is taken in the region outside the lens, see Fig. \ref{fig:The-opening-of lenses}.
By \eqref{eq: variational property of g 1}, \eqref{eq: jump of g}
and \eqref{eq: def of h(z)} we have $2g_{+}\left(\lambda\right)-2\lambda^{2}-l=2h\left(\lambda\right)$,
hence 
\[
\frac{\partial}{\partial\lambda_{0}}\ln H_{n}\left(\lambda_{0},\beta\right)=\frac{\sin\pi\beta}{\pi\sqrt{2n}}\lim_{z\to\lambda}\left(S^{-1}\left(z\right)S'(z)\right)_{21}e^{2nh\left(\lambda\right)}.
\]
Near $\lambda$, we have $S=RP^{(1)}$, and this implies 
\begin{multline}
\frac{\partial}{\partial\lambda_{0}}\ln H_{n}(\lambda_{0},\beta)=\frac{1}{\sqrt{2n}}\frac{1}{\pi}\sin\pi\beta\,\left(\left(P^{(1)}\right)^{-1}\left(P^{(1)}\right)'\right)_{21}(\lambda)\,e^{2nh\left(\lambda\right)}\\
+\frac{1}{\sqrt{2n}}\frac{1}{\pi}\sin\pi\beta\,\left(\left(P^{(1)}\right)^{-1}R^{-1}R'P^{(1)}\right)_{21}(\lambda)\,e^{2nh\left(\lambda\right)}.
\end{multline}
Since $R$ is close to $I$, the second term on the right hand side
is small. Using the asymptotics \eqref{asymptotics R} for $R$, we
obtain 
\begin{equation}
\frac{\partial}{\partial\lambda_{0}}\ln H_{n}(\lambda_{0},\beta)=\frac{1}{\sqrt{2n}}\frac{1}{\pi}\sin\pi\beta\,\left(\left(P^{(1)}\right)^{-1}\left(P^{(1)}\right)'\right)_{21}(\lambda)\,e^{2nh\left(\lambda\right)}+\mathcal{O}\left(n^{-1/2}\right)\,e^{2nh\left(\lambda\right)},
\end{equation}
as $n\to\infty$, uniformly for $t\geq t_{0}$. To compute the remaining
matrix entry, we can use \eqref{eq: P(1) parametrix}, which yields
\begin{multline}
\frac{\partial}{\partial\lambda_{0}}\ln H_{n}(\lambda_{0},\beta)=\zeta'(\lambda)\frac{1}{\sqrt{2n}}e^{-i\pi\beta}\frac{1}{\pi}\sin\pi\beta\,\left(\Psi_{0}^{-1}\Psi_{0}'\right)_{21}(0)\\
+\frac{1}{\sqrt{2n}}\frac{1}{\pi}\sin\pi\beta\,\left(\Phi^{-1}(\tau)E^{-1}(\lambda)E'(\lambda)\Phi(\tau)\right)_{21}(\lambda)+\mathcal{O}\left(n^{-1/2}\right)e^{2nh\left(\lambda\right)},
\end{multline}
as $n\to\infty$. By \eqref{eq: definition of E(z)}, the second term
in the right hand side is of order $\mathcal{O}(n^{-1/6}e^{-\tau})$
uniformly for $t\geq t_{0}$. The first term will be larger than the
last two: by \eqref{eq: expansion zeta}, we get 
\begin{equation}
\frac{\partial}{\partial\lambda_{0}}\log H_{n}(\lambda_{0},\beta)=\sqrt{2}n^{1/6}e^{-i\pi\beta}\frac{1}{\pi}\sin\pi\beta\,\left(\Psi_{0}^{-1}\Psi_{0}'\right)_{21}(0)+\mathcal{O}\left(n^{-1/6}e^{-\tau}\right),
\end{equation}
as $n\to\infty$, uniformly for $t\geq t_{0}$. Write 
\begin{equation}
r(\tau):=\left(\Psi_{0}^{-1}\Psi_{0}'\right)_{21}(0;\tau).\label{def r}
\end{equation}
Then, as $n\to\infty$, 
\begin{equation}
\frac{\partial}{\partial\lambda_{0}}\log H_{n}(\lambda_{0},\beta)=\sqrt{2}n^{1/6}e^{-i\pi\beta}\frac{1}{\pi}\sin\left(\pi\beta\right)\,r(\tau)+\Delta\left(n,t\right),\label{as diff id}
\end{equation}
where $\Delta\left(n,t\right)=\mathcal{O}\left(n^{-1/6}e^{-\tau}\right)$
uniformly for $t\geq t_{0}$ as $n\to\infty$.

\subsection[Expression for $r$ in terms of $u$]{Expression for \boldmath{$r$ in terms of $u$}}
\begin{prop}
Let $r$ be defined by (\ref{def r}), $\Psi_{0}$ as introduced in
Section \ref{sub:Local-parametrix-near-1}, and let $u$ be the Painlevé
II solution characterized by \eqref{u+}. The following identity holds,
\begin{equation}
\frac{\partial}{\partial\tau}r(\tau;\beta)=\frac{-2\pi i}{1-e^{-2i\pi\beta}}u(\tau;\kappa)^{2},\label{id q y}
\end{equation}
where $\kappa$ and $\beta$ are related by \eqref{def beta}. \end{prop}
\begin{proof}
Define 
\begin{equation}
\widehat{\Psi}_{0}(\xi)=\begin{pmatrix}1 & -m_{21}\\
0 & 1
\end{pmatrix}\Psi_{0}(\xi).\label{def hatPsi}
\end{equation}
This transformation has the advantage that it simplifies the $\tau$-equation
in the Lax pair \eqref{eq:Lax system xi}, \eqref{eq:Lax system tau}.
We have 
\begin{equation}
\left(\frac{\partial}{\partial\tau}\widehat{\Psi}_{0}\right)\widehat{\Psi}_{0}^{-1}=-i\xi\begin{pmatrix}0 & 1\\
0 & 0
\end{pmatrix}-i\begin{pmatrix}0 & w\\
-1 & 0
\end{pmatrix},\label{Lax tau}
\end{equation}
where $w$ is some unknown function of $\tau$. In what follows, primes
will be used for differentiation w.r.t. $\tau$.

Now, we start from \eqref{eq: model behavior at zero}. In sector
I, we can write 
\begin{equation}
\Psi_{0}(\xi)=\begin{pmatrix}a & b\\
c & d
\end{pmatrix}(I+E_{1}\xi+\mathcal{O}(\xi^{2}))\left(I+\frac{\kappa^{2}}{2\pi i}\begin{pmatrix}0 & 1\\
0 & 0
\end{pmatrix}\ln\xi\right),\qquad\xi\to0,\label{hatPsi0 expansion}
\end{equation}
for some matrix $E_{1}$ which is $\xi$-independent. We easily see
from (\ref{def r}) and (\ref{def hatPsi}) that 
\begin{equation}
r(\tau)=\left(\Psi_{0}^{-1}\Psi_{0}'\right)_{21}(0;\tau)=\left(\widehat{\Psi}_{0}^{-1}\widehat{\Psi}_{0}'\right)_{21}(0;\tau)=E_{1,21}(\tau).\label{eq q E}
\end{equation}
Substituting (\ref{hatPsi0 expansion}) in (\ref{Lax tau}), we obtain
\begin{equation}
E_{1,21}'(\tau)=ic^{2}(\tau).\label{eq r E2}
\end{equation}
By \eqref{relation c u 1}, we have 
\begin{equation}
E_{1,21}'(\tau)=-\frac{2\pi i}{\kappa^{2}}u(\tau;\kappa)^{2}.
\end{equation}
Together with (\ref{eq q E}) and \eqref{eq r E2}, this implies (\ref{id q y}). 
\end{proof}

\subsection{Proof of Theorem \ref{theorem hankel}}

As $n\to\infty$, we have $\tau\sim t$, see \eqref{eq: expansion tau}.
Integrating (\ref{as diff id}) from $\lambda_{0}=\sqrt{2n}(1+t_{0}n^{-2/3}/2)$
to $\lambda_{1}=\sqrt{2n}(1+t_{1}n^{-2/3}/2)$ with $t_{0}<t_{1}$,
we obtain 
\begin{multline}
\ln H_{n}(\sqrt{2n}(1+\frac{t_{0}}{2}n^{-2/3}),\beta)-\ln H_{n}(\sqrt{2n}(1+\frac{t_{1}}{2}n^{-2/3}),\beta)=\\
-e^{-i\pi\beta}\frac{1}{\pi}\sin\pi\beta\,\int_{t_{0}}^{t_{1}}r(\tau)\mathrm{d}t+\frac{1}{\sqrt{2}}n^{-1/6}\int_{t_{0}}^{t_{1}}\Delta\left(n,t\right)\mathrm{d}t.\label{eq:hankel integral}
\end{multline}
We note that here, as well as in \eqref{as diff id}, 
\[
\tau=\zeta\left(1+\frac{t}{2}n^{-2/3}\right).
\]
Writing the left hand side of this expression in an explicit form,
one can easily check that there exists a positive constant $c_{0}$
such that $\tau\geq c_{0}t$ for all $t>1$ and all $n>1$. Hence
we can let $t_{1}\to+\infty$ in \eqref{eq:hankel integral} and,
taking into account that $e^{-i\pi\beta n}H_{n}\left(\sqrt{2n}\left(1+t_{1}n^{-2/3}/2\right),\beta\right)$
tends to the Gaussian Hankel determinant $H_{n}\left(\lambda_{0},0\right)$
without the jump, arrive at the estimate 
\begin{equation}
\log H_{n}(\sqrt{2n}(1+\frac{t_{0}}{2}n^{-2/3}),\beta)-\log H_{n}\left(\lambda_{0},0\right)-i\pi\beta n=-e^{-i\pi\beta}\frac{1}{\pi}\sin\pi\beta\,\int_{t_{0}}^{\infty}r(\tau)\mathrm{d}\tau+\mathcal{O}\left(n^{-1/3}\right),\label{as Hankel}
\end{equation}
or 
\begin{equation}
H_{n}(\sqrt{2n}(1+\frac{t_{0}}{2}n^{-2/3}),\beta)=-e^{i\pi\beta n}H_{n}\left(\lambda_{0},0\right)\exp\left(-\frac{e^{-i\pi\beta}\sin\pi\beta}{\pi}\,\int_{t_{0}}^{\infty}r(\tau)\mathrm{d}\tau\right)(1+o(1)),
\end{equation}
as $n\to\infty$. We note that we have replaced $\mathrm{d}t$ with
$\mathrm{d}\tau$ in the integral $\int_{t_{1}}^{\infty}r\left(\tau\right)\mathrm{d}\tau$.
This is justified in the limit $n\to\infty$ since $\tau=t\left(1+\mathcal{O}\left(tn^{-2/3}\right)\right)$
and because $r\left(\tau\right)$, being proportional to an integral
of $u\left(\tau;\kappa\right)^{2}$, decays exponentially at positive
infinity. Finally, substituting (\ref{id q y}) into this expression
and integrating by parts (keeping in mind the above mentioned decay
of $r\left(\tau\right)$), we obtain (\ref{hankelas}).

\section*{Acknowledgements}

The work of A. Bogatskiy was supported by the SPbGU grant 11.38.215.2014.
T. Claeys was supported by the European Research Council under the
European Union's Seventh Framework Programme (FP/2007/2013)/ ERC Grant
Agreement n.~307074 and by the Belgian Interuniversity Attraction
Pole P07/18. The work of A. Its was supported in part by NSF grant
DMS-1361856 and the SPbGU grant 11.38.215.2014.

\global\long\def\mkbibnamefirst#1{\textsc{#1}}

\global\long\def\mkbibnamelast#1{\textsc{#1}}

\global\long\def\mkbibnameprefix#1{\textsc{#1}}

\global\long\def\mkbibnameaffix#1{\textsc{#1}}

\SetTracking{encoding=*, shape=sc}{40} 

\printbibliography[heading=bibintoc]

\end{document}